\documentclass[11pt]{scrartcl}

\usepackage{url}
\usepackage[hidelinks]{hyperref}
\usepackage[utf8]{inputenc}
\usepackage[small]{caption}
\usepackage{graphicx}
\usepackage{amsmath}
\usepackage{booktabs}
\usepackage{algorithm,algpseudocode}
\usepackage{amsthm}
\usepackage{amssymb}

\usepackage{cleveref}
\usepackage{color}
\usepackage{xspace}
\usepackage{tikz}
\usepackage{tabularx}
\usepackage{pbox}
\usepackage{framed}
\usepackage{enumitem}   

\DeclareOldFontCommand{\bf}{\normalfont\bfseries}{\mathbf}

\newtheorem{theorem}{Theorem}[section]
\newtheorem{corollary}[theorem]{Corollary}
\newtheorem{proposition}[theorem]{Proposition}
\newtheorem{lemma}[theorem]{Lemma}

\newtheorem{open}[theorem]{Open question}

\theoremstyle{definition}
\newtheorem{definition}[theorem]{Definition}

\newtheorem{remark}[theorem]{Remark}
\newtheorem{claim}{Claim}

\algnotext{EndFor}
\algnotext{EndIf}
\algnotext{EndWhile}
\algnotext{EndProcedure}

\usepackage{natbib}

\usepackage{enumitem}

\allowdisplaybreaks

\definecolor{ForestGreen}{rgb}{.13,.54,.13}

\newcommand{\mms}{\textsc{MMS}}
\newcommand{\wrap}[2]{\textsc{Wrap}(#1,#2)}
\newcommand{\vreq}{V_\text{req}}

\newcommand{\calP}{{\mathcal{P}}}
\newcommand{\cut}{\textsc{Cut}}
\newcommand{\eval}{\textsc{Eval}}
\newcommand{\risn}{\textsc{RISN}}
\newcommand{\gsn}{\textsc{GuSN}}
\newcommand{\gumms}{\textsc{GuMMS}}
\newcommand{\nrect}{\textsc{NRect}}
\newcommand{\ceil}[1]{\lceil #1 \rceil}

\begin{document}

\title{Keep Your Distance: \\Land Division With Separation}

\author{
Edith Elkind\\University of Oxford
\and
Erel Segal-Halevi\\Ariel University
\and
Warut Suksompong\\National University of Singapore
}

\date{\vspace{-3ex}}

\maketitle

\begin{abstract}
This paper is part of an ongoing endeavor to bring the theory of fair division closer to practice by 
handling requirements from real-life applications. We focus on two requirements originating from the 
division of land estates: (1) each agent should receive a plot of a usable geometric shape, and 
(2) plots of different agents must be physically separated. 
With these requirements, the classic fairness notion of 
\emph{proportionality} is impractical, since it may be impossible to attain any multiplicative 
approximation of it. In contrast, the \emph{ordinal maximin share approximation}, introduced by Budish in 
2011, provides meaningful fairness guarantees. We prove upper and lower bounds on achievable maximin 
share guarantees when the usable shapes are squares, fat rectangles, or arbitrary axis-aligned rectangles,
and explore the algorithmic and query complexity of finding fair partitions in this setting.
Our work makes use of tools and concepts from computational geometry such as independent sets of rectangles and guillotine partitions.
\end{abstract}

\section{Introduction}
The problem of fairly allocating a divisible resource has a long history, dating back to the seminal 
article of Polish mathematician Hugo \citet{Steinhaus48}. In its basic formulation, the resource, 
which is metaphorically viewed as a cake, comes in the form of an interval. The aim is to find a 
division satisfying some fairness criteria such as \emph{proportionality}, which means that if there are $n$ 
agents, the value that each of them receives should be at least $1/n$ of the entire cake. 
Not only does a proportional allocation always exist, but it can also be found efficiently \citep{DubinsSp61}.

While the interval cake is simple and consequently useful as a starting point, it is often insufficient 
for modeling real-world applications, especially when combined with the common requirement that each agent 
should receive a connected piece of the cake.\footnote{As \citet{Stromquist80} memorably wrote, without 
such connectivity requirements, there is a danger that agents will receive a `countable union of 
crumbs'.} In particular, when allocating real estate, geometric considerations play a crucial role: it is 
difficult to build a house or raise cattle on a thin or highly zigzagged piece of land even if its total area 
is large. 
Applying algorithms designed for the interval cake may result in an agent receiving, for example, a $5\times 500$ meter land plot, which is hardly useful for any purpose.
Such considerations have motivated researchers to study fairness in land division, which also 
serves to model the allocation of other two-dimensional objects such as advertising spaces 
\citep{BerliantThDu92,IchiishiId99,BerliantDu04,DallAglioMa09,IyerHu09,Devulapalli14,SegalhaleviNiHa17,SegalhaleviHaAu20}. 
These studies have uncovered important differences between land division and interval division: for instance, 
when agents must be allocated square pieces, \citet{SegalhaleviNiHa17} show that we cannot guarantee the 
agents more than $1/(2n)$ of their entire value in the worst case, even when the agents have identical 
valuations over the land.

A related issue, which frequently arises in practice, is that agents' pieces may have to be 
separated from one another: we may need to
leave a space between adjacent pieces of land, e.g., to prevent dispute between owners, 
provide access to the plots, avoid cross-fertilization of different crops, 
or ensure safe social distancing among vendors in a market. The formal study 
of fair division with separation constraints was initiated by \citet{ElkindSeSu22},
who focus on the one-dimensional setting. The goal of our work is to extend 
this analysis to two dimensions, i.e., to analyze
{\em fair division of land under separation constraints}.
In doing so, we will make use of insights and tools from computational geometry such as fat rectangles and guillotine partitions.

We assume that each agent must obtain a contiguous piece of land, 
and the shares that any two agents receive must be separated by 
a distance of at least $s$, where $s$ is a given parameter that is independent of the land value. 
In the presence of separation constraints, \emph{no} multiplicative approximation of proportionality can 
be guaranteed, even in one dimension: when all of the agents' values are concentrated within distance $s$, 
only one agent can receive a positive utility. \citet{ElkindSeSu22} therefore consider the well-known 
criterion of \emph{maximin share fairness} \citep{Budish11,KurokawaPrWa18}---the value that each agent receives 
must be at least her $1$-out-of-$n$ maximin share, i.e., 
the best share that she can guarantee for herself by dividing the resource into $n$ bundles 
and accepting the worst one. \citet{ElkindSeSu22} show that 
this criterion can be satisfied for an interval cake,
while an ordinal approximation of it can be attained for a one-dimensional circular cake.

\subsection{Our contribution}

In this work, we establish that maximin 
share fairness and relaxations thereof can also provide worst-case guarantees in land allocation with 
separation. Moreover, since full proportionality cannot always be attained in this setting even in the 
case of no separation (i.e., $s=0$) due to constraints on the shape of the pieces, our results have interesting implications for that case as well. 

After discussing further related work in \Cref{sec:relatedwork} and introducing the definitions and notation in \Cref{sec:prelim}, 
we prove a general impossibility result in \Cref{sec:impossibility}.
Specifically, we show that when $s>0$, it is impossible to guarantee
to each agent a positive fraction of her $1$-out-of-$n$ maximin share, regardless of the number of agents or the allowed shapes of the pieces. 
Therefore, in the rest
of the paper, we focus on an ordinal notion of approximation.
Specifically, we ask for the smallest value of $k\ge n$ such that we can guarantee
each agent her $1$-out-of-$k$ maximin share, meaning that we allow each agent to divide the land into $k$ parts instead of $n$ parts when computing the maximin share.\footnote{This ordinal relaxation of the maximin share was introduced by \citet{Budish11}, who studied the case $k = n+1$.}

For our positive results, we make the practically common assumption that both the land to be divided and each agent's piece
are axis-aligned rectangles. 
In \Cref{sec:approx}, we show that if all rectangles
(both in agents' maximin partitions and in the final allocation) are required to be \emph{$r$-fat}, 
i.e., the ratio of the length of the longer side to the length of the shorter side is bounded
by a constant $r\ge 1$, then it suffices to set $k=(2\lceil r\rceil+2)n -(3\lceil r\rceil+2)$.
In particular, if all land pieces are required to be squares (i.e., $r=1$), we obtain $k=4n-5$.
Without the fatness assumption, the problem becomes more difficult, 
and the technique we use for fat rectangles does not even yield any finite approximation. 
Nevertheless, we devise a more sophisticated method, based on recursion, that allows us to guarantee a finite approximation: we show that it suffices to set $k=17\cdot 2^{n-3}$, 
and provide stronger bounds for small values of $n$. 
In particular, for $n=2$ we can set $k=3$, which is optimal.

Our positive results in \Cref{sec:approx} are constructive, in the sense that, given each agent's $1$-out-of-$k$
maximin partition (i.e., a partition into $k$ pieces where the value of each piece
is at least the agent's $1$-out-of-$k$ maximin share), 
we can divide the land among the agents so that each agent receives her $1$-out-of-$k$
share, using a natural adaptation of the standard Robertson--Webb model from cake cutting \citep{RobertsonWe98}.
However, it is not clear how a $1$-out-of-$k$ maximin partition can be efficiently computed or even approximated.
To circumvent this difficulty, in \Cref{sec:guillotine}, we focus on a special class of land partitions known as 
{\em guillotine partitions}; 
intuitively, these are partitions that can be obtained by a sequence of edge-to-edge cuts.\footnote{Guillotine partitions have been extensively studied in computational geometry; we refer to the beginning of \Cref{sec:guillotine-definition} for some references.} 
We show that we can efficiently compute an approximately optimal guillotine partition,
and that the loss caused by using guillotine partitions compared to arbitrary partitions can be bounded. 
Combining
these results with our ordinal approximation algorithms from \Cref{sec:approx}, we obtain
approximation algorithms for computing a maximin allocation.

\section{Related Work}
\label{sec:relatedwork}

\subsection{Geometric division}
Division problems are abundant in computational geometry. A survey by \citet{keil2000polygon}
lists over 100 papers about different variants of
such problems. 
A typical problem involves a given polygon $C$ and a given family $U$ of polygons (e.g., triangles, squares, rectangles, stars, or spirals). 
The polygon $C$ is to be partitioned into a number of parts, each of which belongs to the family $U$. 
The partition
should satisfy requirements such as  minimizing the
number of pieces or minimizing the total perimeter of
the pieces. 
Sometimes it is also required that the
pieces have the same area as well as the same perimeter \citep{nandakumar2012fair,blagojevic2014convex,armaselu2015algorithms,frettloh2022fair}.

However, the value
of land is much more than its shape, area, or perimeter. 
For
example, a land plot near the sea may have a much higher value than one with exactly the same
shape and area in the middle of a desert. Geometric
partition problems do not handle such considerations.

\subsection{Additional related work}

In considering fair division with separation, we build on the work of \citet{ElkindSeSu22},
who investigate the one-dimensional variant of this problem. 
Fair land division with constraints on the shape of usable pieces
has been previously studied \citep{SegalhaleviNiHa17,SegalhaleviHaAu20}.
We follow these works in considering fat rectangles and guillotine cuts;
however, the fairness notions considered in these papers are (partial) proportionality
and envy-freeness, whereas our work concerns maximin fairness. 
As we noted earlier, under separation constraints, it can happen that all agents but one receive zero utility in every allocation, which renders proportionality
and envy-freeness infeasible.

Our analysis 
is also somewhat similar in spirit
to a number of recent works on dividing a cake represented by a general graph, which generalizes both the interval and the cycle (a.k.a.~pie) setting. 
Several fairness notions have been studied in that setting: partial proportionality \citep{BeiSu21}, envy-freeness \citep{IgarashiZw21}, and maximin share fairness \citep{ElkindSeSu21-Graph}.
In all of those works, the cake is still one-dimensional: it is a union of a finite number of intervals. 
As we show in this work, a two-dimensional cake is fundamentally different.
Our work fits into the ongoing endeavor to bring the theory of fair division closer to practice by 
handling constraints from real-life applications---\citet{Suksompong21} surveys other types of constraints that have been investigated in the fair division literature.

Finally, we remark that we impose the geometric and separation requirements not only on the final allocation but also in the definition of the maximin share benchmark.
This is consistent with prior uses of the maximin share under constraints \citep{BouveretCeEl17,BiswasBa18,LoncTr20,ElkindSeSu21-Graph,ElkindSeSu22,BeiIgLu22}.

\section{Preliminaries}
\label{sec:prelim}

The {\em land} is given by a closed, bounded, and connected subset $L$ of the two-dimensional Euclidean plane $\mathbb{R}^2$.
The land~$L$ is to be divided among a set of agents
${\mathcal N} = [n]$, where $[k] := \{1,2,\dots,k\}$ for any positive integer $k$.
There is a prespecified family $U$ of \emph{usable} pieces, for example, squares or `fat' rectangles.
We require that usable pieces are connected, and that if a piece belongs to~$U$, then any piece that results from expanding or shrinking it by an arbitrary factor also belongs to $U$.
Each agent has an integrable \emph{value-density function} $f_i:L\rightarrow \mathbb{R}_{\ge 0}$.
For each $i\in\mathcal{N}$, agent $i$'s value for a piece of land~$Z$ is given by $v_i(Z) := \iint_Z f_i(x,y) dxdy$.
An \emph{instance} consists of the land, agents, density functions, and the \emph{separation parameter} $s
\geq 0$.

An \emph{allocation} of the land is given by a list $\mathbf{A}=(A_1,\dots,A_n)$, where each $A_i$ is a single usable piece of land allocated to agent~$i$.
We require allocations to be \emph{$s$-separated}, i.e., any two pieces $A_i$ and $A_j$ are separated by distance at least $s$, where distance is measured according to the $\ell_\infty$ norm:\footnote{Distances with respect to different norms are not far from one another.
For instance, the Euclidean $\ell_2$ distance differs from the $\ell_\infty$ distance by only a factor of $\sqrt{2}$: $d_\infty(\mathbf{p},\mathbf{q})\leq d_2(\mathbf{p},\mathbf{q})\leq \sqrt{2}\cdot d_\infty(\mathbf{p},\mathbf{q})$ for any points $\mathbf{p}$ and $\mathbf{q}$.}
\begin{align*}
d(A_i,A_j) 
= \inf_{(x,y)\in A_i, (x',y')\in A_j}\max\{|x-x'|,|y-y'|\}.
\end{align*}
\emph{Partitions} and \emph{$s$-separated partitions} are defined similarly, except that instead of a list $\mathbf{A}=(A_1,\dots,A_n)$, we have a set $\mathbf{P}=\{P_1,\dots,P_n\}$.\footnote{With a slight abuse of language, we use the term `partition' for a set of pairwise-disjoint pieces, even though some land remains unallocated due to the separation and geometric constraints.
}
Denote by $\Pi_n(s)$ the set of all $s$-separated partitions into $n$ pieces.

We now define the main fairness notion of interest in this paper.

\begin{definition}
\label{def:MMS}
For any positive integer $k$, the $1$-out-of-$k$ \emph{maximin share} of agent $i$ is defined as 
\[
\mms_i^{k,s} := \sup_{\mathbf{P}\in \Pi_k(s)}\min_{j\in [k]}v_i(P_j).
\]
We omit $s$ if it is clear from the context, and write $\mms_i^k$ instead of $\mms_i^{k,s}$.
Further, we refer to $\mms_i^n$ as $i$'s {\em maximin share}. 
\end{definition}

For instance, suppose that the land is a unit square, the family of usable pieces consists of all rectangles, and an agent has value~$1$ for the entire square, spread uniformly over the square.
Then, for any $s\le 1$, the agent's $1$-out-of-$2$ maximin share is $(1-s)/2$.

As with cake cutting \citep{ElkindSeSu22}, the supremum in \Cref{def:MMS} can be replaced by a maximum.\footnote{This requires defining a metric on the usable pieces and showing that $U$ is compact in that metric space---see Appendix C in the work of
\citet{SegalhaleviNiHa17}.}
An $s$-separated partition for which this maximum is attained is called a ($1$-out-of-$k$) \emph{maximin partition} of agent~$i$.

\section{A General Impossibility Result}
\label{sec:impossibility}

We begin by showing that, in contrast to one-dimensional cake cutting with separation \citep{ElkindSeSu22}, for land division there may be no allocation that guarantees to all agents their maximin share or even any multiplicative approximation of it.
This negative result does not depend on the geometric shape of the land or the usable pieces.

The proof of this result uses the following geometric lemma.

\begin{lemma}
\label{lem:points}
For every integer $n\geq 2$, 
there exist $n$ sets $S_1, \dots, S_n$ consisting of $n$ points each, 
such that 

(a) the $\ell_\infty$ distance between any two points in the same set is greater than $1$;

(b) if we pick one representative point from each set, then some two representatives are at $\ell_\infty$ distance less than $1$ apart.
\end{lemma}
\begin{proof}\footnote{
We are grateful to Alex Ravsky for the proof idea. 
}
We measure distance in $\ell_\infty$ norm unless specified otherwise.
Consider a regular $2n$-gon with side length $1$ (in $\ell_2$ norm)\footnote{Recall that in $\ell_2$ norm, the distance between points $(x,y)$ and $(x',y')$ is $\sqrt{(x-x')^2 + (y-y')^2}$.} and vertices $1,2,\dots,2n$ in this order, and position it on the plane so that none of its edges is parallel to either the $x$-axis or the $y$-axis; in particular, for $n=2$,
position it so that one of its vertices is its rightmost point and the two vertices adjacent to it are aligned vertically.

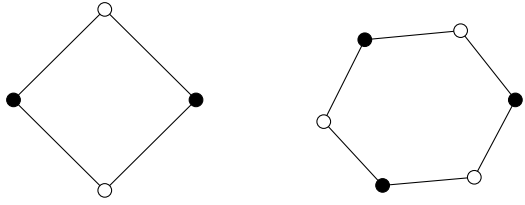
\begin{figure}
\centering
\begin{tikzpicture}[scale=0.6]

\draw (4,2) -- (2,4) -- (0,2) -- (2,0) -- (4,2);
\draw[fill=black] (4,2) circle [radius = 0.15];
\draw[fill=black] (0,2) circle [radius = 0.15];
\draw[fill=white] (2,4) circle [radius = 0.15];
\draw[fill=white] (2,0) circle [radius = 0.15];

\draw (11,2) -- (9.8,3.53) -- (7.7,3.33) -- (6.8,1.52) -- (8.09,0.11) -- (10.1,0.29) -- (11,2);
\draw[fill=black] (11,2) circle [radius = 0.15];
\draw[fill=white] (6.8,1.52) circle [radius = 0.15];
\draw[fill=white] (9.8,3.53) circle [radius = 0.15];
\draw[fill=black] (7.7,3.33) circle [radius = 0.15];
\draw[fill=black] (8.09,0.11) circle [radius = 0.15];
\draw[fill=white] (10.1,0.29) circle [radius = 0.15];
\end{tikzpicture}
\caption{Construction in the proof of \Cref{lem:points} for $n=2$ (left) and $n=3$ (right). Each of the first $n-1$ sets consists of the black vertices while the $n$-th set consists of the white vertices.
All sides are of length $1$.
}
\label{fig:impossibility}
\end{figure}
Let each of $S_1, \dots, S_{n-1}$ consist of all odd-numbered vertices, and let $S_n$ consist of all even-numbered vertices (see \Cref{fig:impossibility}).

If $n=2$, then the $\ell_\infty$ distance between two vertices from the same set is $\sqrt{2}$, while the $\ell_\infty$ distance between a vertex in $S_1$ and a vertex in $S_2$ is $1/\sqrt{2}$. 

Now, suppose $n\ge 3$, and consider any set of representatives.
If two representatives coincide, the distance between them is $0$.
Else, the representatives of $S_1, \dots, S_{n-1}$ are all distinct, so no matter which representative we choose for $S_n$, there will be two representatives that correspond to adjacent vertices of the $2n$-gon.
Since the edge connecting these two representatives is not axis-parallel, the $\ell_\infty$ distance between them is less than~$1$. 

It remains to argue that the distance between any two points from the same set is more than $1$. 
Fix two such points $a$ and $b$.
Note that for $n>2$ the angles of the regular $2n$-gon
are obtuse. 
Hence, if $a$ and $b$ have exactly one vertex of the $2n$-gon between them, by the law of cosines the $\ell_2$ distance between $a$ and $b$ is greater than $\sqrt{2}$; otherwise, the $\ell_2$ distance is even higher. 
It follows that the $\ell_\infty$ distance between $a$ and~$b$ is greater than $1$. 
Hence, the sets $S_1,\dots,S_n$ satisfy the claimed properties.
\end{proof}

\begin{theorem}
\label{thm:impossibility}
For every
family $U$ of usable pieces,
integer $n\geq 2$, separation parameter $s>0$, and real number $r>0$, there exists a land division instance with $n$ agents in which no $s$-separated allocation gives every agent $i$ a value of at least $r\cdot\emph{MMS}_i^{n,s}$.
\end{theorem}

\begin{proof}
Let $S_1, \dots, S_n$ be sets as in \Cref{lem:points}.
Let $\delta := \delta(n) > 0$ be such that the $\ell_\infty$ distance between any two points in the same set is at least $1+\delta$.
We construct value-density functions for the $n$ agents as follows. 

Pick a positive 
$\varepsilon< \frac{s\delta}{6+4\delta}$, and 
scale the axes so that whenever we pick one representative from each set, one of the pairwise distances among them is less than $s-4\varepsilon$,  while the distance between any two points in the same set is at least $(s-4\varepsilon)(1+\delta)$.
Since $6\varepsilon + 4\varepsilon\delta< s\delta$, we have $\delta > \frac{6\varepsilon}{s-4\varepsilon}$, and therefore $(s-4\varepsilon)(1+\delta) > (s-4\varepsilon)\cdot\frac{(s-4\varepsilon)+6\varepsilon}{s-4\varepsilon} = s+2\varepsilon$.
For each agent $i$, we create $n$ small `pools' (i.e., usable shape $S\in U$ such that any two points in the same pool are at most $\varepsilon$ apart) around points in~$S_i$, so that the agent values each pool at $1$. 
Since the distance between any two points in $S_i$ is at least $s+2\varepsilon$, the triangle inequality implies that the distance between any two pools is at least $s$, so the maximin share of each agent is $1$.

Suppose for contradiction that there exists an allocation in which all agents receive a positive fraction of their maximin share. 
This means that every agent $i$ receives a piece that overlaps at least one pool around a point in $S_i$. 
By construction of $S_1, \dots, S_n$, there are two agents $i$ and $j$ such that $i$ is allocated
a piece that overlaps a pool with center $x\in S_i$, $j$ is allocated a piece that
overlaps a pool with center $y\in S_j$, and the distance between $x$ and $y$ is less than $s-4\varepsilon$.
Since the diameter of each pool is at most~$\varepsilon$, 
the triangle inequality implies that the distance between $i$'s piece and $j$'s piece 
is at most $\varepsilon + (s-4\varepsilon) + \varepsilon < s$, i.e., 
these pieces are not $s$-separated.
This yields the desired contradiction. 
\end{proof}

\section{Ordinal Approximation}
\label{sec:approx}

Since no multiplicative approximation of the maximin share can be guaranteed, we instead consider an 
ordinal notion of approximation. That is, we ask if each agent can be guaranteed her
$1$-out-of-$k$ maximin share for some  $k > n$. 

While the negative result of Section~\ref{sec:impossibility} 
does not depend on geometric shape constraints, our positive 
results concern pieces that have a `nice' geometric shape. Specifically, in \Cref{sec:square-fat}, we first consider the scenario
where the set $U$ of usable pieces consists of `fat' axis-aligned rectangles, i.e., rectangles whose length-to-width
ratio is bounded by a constant (for example, if this constant is $1$, the set~$U$ consists
of axis-aligned squares). 
We then consider in \Cref{sec:arbitrary-rect} the more general setting where $U$ consists of all axis-aligned
rectangles. 
Placing such constraints on the shape of each piece is useful in land 
allocation settings, particularly in urban regions. Note that when we restrict the shape of the usable 
pieces to be a (fat) rectangle, in our definition of the maximin share we also only consider 
$s$-separated partitions in which each piece is a (fat) rectangle.

\subsection{Squares and fat rectangles}
\label{sec:square-fat}

Given a rectangle $R$, we denote the lengths 
of its long and short side by $\text{long}(R)$ 
and $\text{short}(R)$, respectively.
For any real number $r\ge 1$, a rectangle $R$ is called \emph{$r$-fat} if 
${\text{long}(R)}/{\text{short}(R)}\leq r$
\citep{AgarwalKaSh95,Katz97,SegalhaleviNiHa17,SegalhaleviHaAu20}.
In particular, a $1$-fat rectangle is a square.

In order to obtain maximin share guarantees, the high-level idea is to find a sufficiently large $k$ such 
that if we consider the agents' $1$-out-of-$k$ maximin partitions, then it is always possible to select a 
representative piece from each partition in such a way that these representatives are $s$-separated. This 
will ensure that, by allocating to each agent her representative, we obtain an allocation that is 
$s$-separated and in which agent $i$ receives value at least $\mms_i^{k,s}$. 
The following 
theorem shows that $k=(2\lceil r\rceil + 2)n - (3\lceil r\rceil + 2)$ suffices; 
for a square ($r=1$), this yields $k = 4n-5$. 
Note that our result does not place any assumptions on the shape of the land.

In what follows, we say that two pieces of land \emph{overlap} if their intersection 
has a positive area, and that they are \emph{disjoint} if their intersection has an area of zero.

\begin{definition}
Given $n$ sets of geometric objects, $S_1,\ldots,S_n$, where the objects in each set are pairwise-disjoint, a \emph{rainbow independent set} is a selection of a single representative object from each set, such that the $n$ representatives are pairwise-disjoint.
\end{definition}
Rainbow independent sets have been studied in abstract graphs \citep{aharoni2019RIS,lv2021rainbow,kim2022rainbow,ma2022rainbow}; here, we are interested in independent sets in the \emph{intersection graph}, that is, the graph in which the vertices are the rectangles and the edges connect pairs of overlapping rectangles.
Independent sets in intersection graphs have been studied for various families of shapes, particularly for axis-parallel rectangles \citep{agarwal1998label,chan2004note,chalermsook2009maximum,adamaszek2013approximation,chuzhoy2016approximating,chalermsook2021coloring,galvez20223,mitchell2022approximating},
but without the `rainbow' constraint.

Naturally, when the $n$ sets are larger, it is more likely that a rainbow independent set exists. We are interested in the smallest $k$ (in terms of $n$) such that a rainbow independent set exists for every collection of $n$ sets of size $k$ each.

\begin{definition}
Let $U$ be a family of geometric objects and $n \geq 2$ an integer. The \emph{$n$-th rainbow-independent-set number} of $U$ is the smallest integer $k$ such that, for every collection of $n$ sets, each consisting of $k$ pairwise-disjoint objects from the family $U$, a rainbow independent set exists.

For every real number $r\geq 1$ and integer $n\geq 2$, denote by $\risn(r,n)$ the $n$-th rainbow-independent-set number when $U$ is the family of $r$-fat rectangles.
\end{definition}

\begin{lemma}
\label{lem:rfat2}
For every real number $r\geq 1$, 
$\risn(r,2)=\ceil{r}+2$.
\end{lemma}
\begin{proof}
\footnote{We are grateful to an anonymous Mathematics Stack Exchange contributor for the proof idea.
}
We first show that
$\risn(r,2)\leq \lceil r\rceil+2$.
To this end, given $\lceil r\rceil+2$ red and $\lceil r\rceil+2$ blue $r$-fat rectangles, we have to show that at least one red and one blue rectangle do not overlap each other.

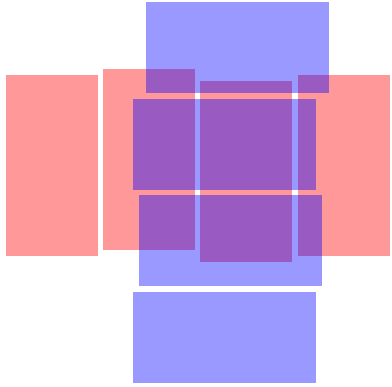
\begin{figure}
\centering
\begin{tikzpicture}[scale=0.8]
\path [fill=red, opacity=0.4] (1.5,2.5) rectangle (3,5.5);
\path [fill=red, opacity=0.4] (3.1,2.6) rectangle (4.6,5.6);
\path [fill=red, opacity=0.4] (4.7,2.4) rectangle (6.2,5.4);
\path [fill=red, opacity=0.4] (6.3,2.5) rectangle (7.8,5.5);
\path [fill=blue, opacity=0.4] (3.6,0.4) rectangle (6.6,1.9);
\path [fill=blue, opacity=0.4] (3.7,2) rectangle (6.7,3.5);
\path [fill=blue, opacity=0.4] (3.6,3.6) rectangle (6.6,5.1);
\path [fill=blue, opacity=0.4] (3.8,5.2) rectangle (6.8,6.7);
\end{tikzpicture}
\caption{Illustration of the proof of \Cref{lem:rfat2} for $r=2$.
The blue rectangles lie in order vertically and the red rectangles lie in order horizontally.}
\label{fig:r-fat}
\end{figure}

Assume for contradiction that every red rectangle overlaps all blue rectangles (and vice versa). 
Consider any two red rectangles.
If their projections on the $x$-axis overlap, then, since the rectangles are disjoint, 
their projections on the $y$-axis do not overlap.
Hence, we can rotate the plane so that one of the red rectangles (say, Red$_1$) lies entirely to the left of another red rectangle (say, Red$_2$).
Since all blue rectangles overlap both Red$_1$ and Red$_2$, their projections on the $x$-axis must contain the projection of the space between the right side of Red$_1$ and the left side of Red$_2$. 
Now, since the blue rectangles are pairwise disjoint, they must lie
in order vertically. Similar considerations imply that the red rectangles must lie in order horizontally; see Figure \ref{fig:r-fat}.

In each family, there are two extreme rectangles (top and bottom blue rectangle; leftmost and rightmost red rectangle); we call the other $\lceil r\rceil$ rectangles in each family \emph{middle rectangles}.
Let Red$_{\min}$ (resp., Blue$_{\min}$) be 
a red (resp., blue) rectangle with a smallest shorter side among all middle red (resp., blue) rectangles. 
Since every red rectangle overlaps all blue rectangles,
Red$_{\min}$ overlaps the $\lceil r\rceil$ middle blue rectangles and the space above and below them.
Hence, the longer side of Red$_{\min}$ is strictly longer than the shorter sides of these $\lceil r\rceil$ middle blue rectangles combined. 
This means that
\begin{align*}
\text{long}(\text{Red}_{\min})
>
\lceil r\rceil \cdot 
\text{short}(\text{Blue}_{\min}).
\end{align*}
By a similar argument, we have
\begin{align*}
\text{long}(\text{Blue}_{\min})
>
\lceil r\rceil \cdot 
\text{short}(\text{Red}_{\min}).
\end{align*}
Moreover, by the definition of fatness,
\begin{align*}
r \cdot 
\text{short}(\text{Blue}_{\min})
&\geq
\text{long}(\text{Blue}_{\min});
\\
r \cdot 
\text{short}(\text{Red}_{\min})
&\geq
\text{long}(\text{Red}_{\min}).
\end{align*}
Combining the four inequalities gives $\text{long}(\text{Red}_{\min}) > \text{long}(\text{Red}_{\min})$, a contradiction.
Hence, a rainbow independent set exists, so $\risn(r,2)\leq \lceil r\rceil+2$.

To show that $\risn(r,2) \geq \lceil r\rceil+2$,
we exhibit two sets of $r$-fat rectangles, each of which contains 
$\ceil{r}+1$ pairwise-disjoint rectangles, such that no two representatives are disjoint.
Consider the following two sets of rectangles.
The \emph{vertical} set contains the rectangles  $[i,i+1]\times[1-\varepsilon,~\ceil{r}+\varepsilon]$
for $i\in\{0,\ldots, \ceil{r}\}$, for some $\varepsilon > 0$.
Analogously, the \emph{horizontal} set contains the rectangles $[1-\varepsilon,~\ceil{r}+\varepsilon]\times [i,i+1]$
for $i\in\{0,\ldots, \ceil{r}\}$.
For all rectangles, the ratio between the two side lengths is $\ceil{r}+2\varepsilon-1$; when $\varepsilon$ is sufficiently small, this ratio is less than $r$, so all rectangles are $r$-fat. 
It can be verified that each horizontal rectangle overlaps all vertical rectangles.
\end{proof}

\begin{lemma}
\label{lem:rfatn}
For every real number $r\geq 1$ and integer $n\ge 2$,
\begin{align*}
\risn(r,n)
\leq 
(2\lceil r\rceil + 2)n - (3\lceil r\rceil  + 2).
\end{align*}
\end{lemma}
\begin{proof}
We first prove that 
$\risn(r, n+1)\leq \risn(r, n) + (2\lceil r\rceil + 2)$
for every $r\geq 1$ and $n\ge 2$. 
To this end, we claim that in any collection $\mathcal{C}$ of axis-aligned $r$-fat rectangles, there exists a rectangle in $\mathcal{C}$ that overlaps at most $2\lceil r\rceil + 2$ pairwise-disjoint rectangles from $\mathcal{C}$.
To prove this claim, let $Q_{\min}$ be a rectangle with a smallest shorter side among all rectangles in $\mathcal{C}$, and denote the length of this side by $w$. 
Mark all four corners of $Q_{\min}$.
Additionally, on each of its two longer sides,\footnote{If $Q_{\min}$ is a square, choose its `longer' sides arbitrarily.} mark $\lceil r\rceil - 1$ additional points so that the distance between any two consecutive marks (including the two marks at the corners) is at most $w$; this is possible because $Q_{\min}$ is $r$-fat, so $\text{long}(Q_{\min})\leq r\cdot w$.
The total number of marks is $2\lceil r\rceil + 2$.

Assume for contradiction that at least $2\lceil r\rceil + 3$ pairwise-disjoint rectangles in $\mathcal{C}$ overlap $Q_{\min}$; call the set of these pairwise-disjoint rectangles $\mathcal{D}$.
Let $a$ be the number of pairs $(Q,m)$ such that the mark $m\in Q_{\min}$ is contained in the interior of the rectangle $Q\in \mathcal{D}$, and let $b$ be the number of pairs $(Q,m)$ such that the mark $m\in Q_{\min}$ is on the border of the rectangle $Q\in \mathcal{D}$.
Any rectangle $Q\in \mathcal{C}$ that overlaps $Q_{\min}$ must either contain at least one mark of $Q_{\min}$ in its interior, or have at least two marks of $Q_{\min}$ on its border and overlap the area of $Q_{\min}$ adjacent to these marks.
Hence, we have $a + b/2 \ge 2\lceil r\rceil + 3$.
On the other hand, since the rectangles in $\mathcal{D}$ are pairwise disjoint, each mark of $Q_{\min}$ can only be contained in one rectangle from $\mathcal{D}$ (in which case it cannot be on the border of any other rectangle from~$\mathcal{D}$) or on the border of one or two rectangles from $\mathcal{D}$ (in which case it cannot be contained in any rectangle from $\mathcal{D}$).
Since the number of marks is $2\lceil r\rceil + 2$, it holds that $a + b/2 \le 2\lceil r\rceil + 2$.
This yields a contradiction and establishes the claim.

Now, suppose we are given $n+1$ sets, each of which contains $\risn(r, n) + (2\lceil r\rceil + 2)$ pairwise-disjoint rectangles.
By our claim, among all of these rectangles, there exists a rectangle~$R$ that overlaps at most $2\lceil r\rceil + 2$ pairwise-disjoint rectangles.
Remove $R$, the set that contains $R$, and all rectangles from the remaining $n$ sets that overlap $R$.
Each of the remaining $n$ sets still contains at least $\risn(r, n)$ rectangles, so it is possible to choose pairwise-disjoint representatives from these sets.
Moreover, these representatives do not overlap $R$, so we may choose $R$ as the representative of its set.
It follows that $\risn(r, n+1)\leq \risn(r, n) + (2\lceil r\rceil + 2)$.

Combining this inequality with \Cref{lem:rfat2} gives
\begin{align*}
\risn(r,n) &
\leq 
(2\lceil r\rceil + 2)(n-2) + (\lceil r\rceil + 2)
= 
(2\lceil r\rceil + 2)n - (3\lceil r\rceil  + 2),
\end{align*}
as claimed.
\end{proof}
For the special case of squares we get $\risn(1,n)\leq 4n-5$.
Constructions similar to those in the proof of \Cref{thm:impossibility} show that\footnote{Specifically, we place same-size axis-parallel squares centered at each vertex of the $2n$-gon in the proof of \Cref{thm:impossibility}, with the side length of the squares chosen so that any pair of squares centered at adjacent vertices of the $2n$-gon overlap, while any pair of squares at non-adjacent vertices do not.} $\risn(r,n)\geq n+1$ for all~$r$.
Thus the 
bound $4n-5$ for squares is optimal for $n=2$, but may be suboptimal for $n\geq 
3$. 
\begin{open}
What is the exact value of $\risn(r,n)$ for every $r\geq 1$ and $n\geq 3$?
In particular, what is the $n$-th rainbow independent set number of squares for $n\geq 3$?
\end{open}

We can finally return to our original problem.

\begin{theorem}
\label{thm:r-fat}
Let $r\ge 1$ be a real number.
For every land division instance with $n\ge 2$ agents, separation parameter $s$, and $U$ being the set of axis-aligned $r$-fat rectangles, 
there exists an $s$-separated allocation in which every agent $i$ receives value at least 
$\emph{MMS}_i^{k, s}$, where $k := (2\lceil r\rceil + 2)n - (3\lceil r\rceil  + 2)$.
\end{theorem}

\begin{proof}
Given the separation parameter $s$, for every rectangle $Q$ with side lengths $u$ and~$v$, define $\wrap{Q}{s}$ to be the rectangle with side lengths $u + s$ and $v + s$ and the same center as $Q$ (i.e., wrap $Q$ with a `rectangle ring' of width $s/2$).
Observe that if $Q$ is $r$-fat, then so is $\wrap{Q}{s}$.
Moreover, two rectangles $Q_1$ and $Q_2$ are $s$-separated if and only if $\wrap{Q_1}{s}$ and $\wrap{Q_2}{s}$ do not overlap (see Figure~\ref{fig:square-wrap}).

\begin{figure}
\centering
\begin{tikzpicture}[scale=0.9]
\path [fill=blue!30] (0,0) rectangle (3,3);
\path [fill=blue] (0.5,0.5) rectangle (2.5,2.5);
\path [fill=red!30] (3.1,0.3) rectangle (6.1,2.3);
\path [fill=red] (3.6,0.8) rectangle (5.6,1.8);
\end{tikzpicture}
\caption{The dark rectangles are $s$-separated (in the $\ell_{\infty}$ metric) if and only if the light rectangles wrapping them with a rectangle ring of width $s/2$ are disjoint.}
\label{fig:square-wrap}
\end{figure}
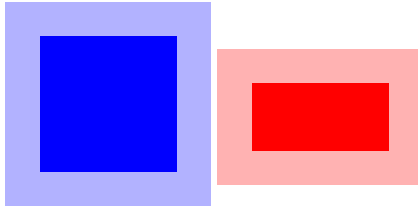

Now, given an $n$-agent instance, we ask each agent to produce a $1$-out-of-$k$ maximin partition: 
this is a set of $k$ axis-aligned rectangles that are $s$-separated. 
Then, we replace each rectangle $Q$ with $\wrap{Q}{s}$, so each agent now has a set of 
$k$ pairwise-disjoint rectangles.
\Cref{lem:rfatn} implies that
$k\ge \risn(r,n)$, so 
the collection of partitions admits a rainbow independent set.
Denote the elements of this set by $\wrap{Q_1}{s},\allowbreak \wrap{Q_2}{s},\allowbreak \dots,\allowbreak \wrap{Q_n}{s}$, 
where $\wrap{Q_i}{s}$ belongs to agent $i$'s partition. 
We allocate the rectangle $Q_i$ to agent $i$.
The rectangles $Q_1,\dots,Q_n$ are $s$-separated, and every agent $i$ receives value at least $\mms_i^{k, s}$, as desired.
\end{proof}

\subsection{Arbitrary rectangles}
\label{sec:arbitrary-rect}

Next, we allow the allocated pieces to be arbitrary axis-aligned rectangles, and assume that the land itself is also an axis-aligned rectangle.
Without loss of generality, we suppose further that the land is a square (otherwise, for positive results, a rectangular land can be completed to a square by attaching to it a rectangle that all agents value at~$0$).
We scale the axes so that the land is the unit square $[0,1]\times[0,1]$.

The arbitrary rectangle case differs from the fat rectangle case in two respects.
First, \emph{without} the separation requirement, the arbitrary rectangle case is much \emph{easier}:
the land can be projected onto a one-dimensional interval, for which full proportionality, and hence $\mms_i^n$,
can be achieved \citep{DubinsSp61}.
In contrast, \emph{with} the separation requirement, the arbitrary rectangle case is much \emph{harder}: 
the rainbow-independent-set technique of \Cref{thm:r-fat} 
does not yield a meaningful bound for arbitrary rectangles.
Indeed, \Cref{lem:rfat2} implies that the rainbow-independent-set number of rectangles is infinite even for $n=2$
(that is, 
even for an arbitrarily large $k$, there exist two size-$k$ sets 
of pairwise-disjoint rectangles such that no two representatives are disjoint).

A priori, for $n\geq 2$, it is not even clear that there is a finite $\nrect(n)$ such that an 
$\mms^{\nrect(n)}$ allocation among $n$ agents always exists.
We shall prove that $\nrect(n)$ is indeed finite for any $n\geq 2$, 
and derive improved upper bounds on $\nrect(n)$ for small values of $n$.
Towards this goal, we develop some new tools.

In what follows, for each agent we fix a $1$-out-of-$k$ maximin 
partition; see \Cref{fig:rectangle-2-3} for some examples of such partitions when $k = 3$.
For all $i\in\mathcal N$,
we assume without loss of generality that $\mms_i^k = 1$, and that $i$'s 
value is $0$ outside the $k$ rectangles in her maximin partition 
(the latter value being positive can only make it easier to satisfy the agent). 
We may also assume that $i$'s value for each rectangle in her partition is exactly~$1$.
Hence each agent has a value of exactly $k$ for the land and should get an axis-aligned rectangle worth at least $1$. 

We refer to the $k$ rectangles in an agent's (fixed) maximin partition 
as \emph{MMS-rectangles}; every rectangular piece of land that is worth 
at least $1$ to the agent is called a \emph{value-$1^+$ rectangle}. 
Due to our normalization, every MMS-rectangle is a value-$1^+$ rectangle, but the converse is not necessarily true.

\begin{figure*}
\centering
\begin{tikzpicture}[scale=0.9]

\node at (1.5,1.5) {\footnotesize (a)};
\draw (0,2) rectangle +(3,3);
\draw[fill=green] (0,2) rectangle +(1,3);
\draw[fill=green] (1.1,2) rectangle +(0.5,3);
\draw[fill=green] (1.7,2) rectangle +(1.3,3);

\node at (5.7,1.5) {\footnotesize (b)};
\draw (4.2,2) rectangle +(3,3);
\draw[fill=green] (4.2,2) rectangle +(3,1.6);
\draw[fill=green] (4.2,3.7) rectangle +(2,1.3);
\draw[fill=green] (6.3,3.7) rectangle +(0.9,1.3);

\node at (9.9,1.5) {\footnotesize (c)};
\draw (8.4,2) rectangle +(3,3);
\draw[fill=green] (8.4,2) rectangle +(3,1);
\draw[fill=green] (8.4,3.1) rectangle +(3,1.2);
\draw[fill=green] (8.4,4.4) rectangle +(3,0.6);

\node at (14.1,1.5) {\footnotesize (d)};
\draw (12.6,2) rectangle +(3,3);
\draw[fill=green] (12.6,2) rectangle +(1.6,2);
\draw[fill=green] (12.6,4.1) rectangle +(1.6,0.9);
\draw[fill=green] (14.3,2) rectangle +(1.3,3);

\end{tikzpicture}
\caption{Four partitions of a rectangular land into three axis-aligned rectangles. Each of these partitions can be a $1$-out-of-$3$ maximin partition of an agent.  
The white space between the green rectangles has thickness $s$.}
\label{fig:rectangle-2-3}
\end{figure*}
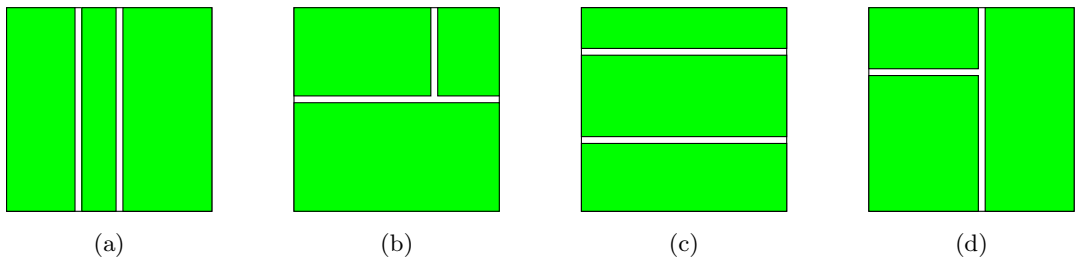

\begin{definition}
Consider an agent with a fixed $1$-out-of-$k$ maximin partition, and integers $p,q\geq 1$.
A \emph{vertical $p$:$q$-rectangle cut} is a rectangular strip
of height $1$ and width $s$ that has at least 
$p$ whole MMS-rectangles on its left and at least $q$ whole MMS-rectangles on its right.
A \emph{vertical $p$-rectangle stack} is 
a sequence of $p$ value-$1^+$ rectangles
such that each consecutive pair is separated by a vertical distance of at least~$s$.
\end{definition}
In Figure \ref{fig:rectangle-2-3}(a), the left vertical cut is a vertical $1$:$2$-rectangle cut and the right one is a vertical $2$:$1$-rectangle cut.
In Figure \ref{fig:rectangle-2-3}(c), there is a vertical $3$-rectangle stack.
In Figure \ref{fig:rectangle-2-3}(d), the vertical cut is a $2$:$1$-rectangle cut, and there is a vertical $2$-rectangle stack.

The following lemma shows the existence of either a rectangle cut or a rectangle stack with appropriate parameters.

\begin{lemma}
\label{lem:cuts}
Fix an agent and a $1$-out-of-$k$ maximin partition of this agent.
For any integers $1\leq p,q \leq k$ with $p+q \leq k+1$, the agent
has a vertical $p$:$q$-rectangle cut or a vertical $(k-p-q+2)$-rectangle stack.
\end{lemma}

\begin{proof}
Starting from the left end of the cake, move a vertical knife of width $s$ to the right.
Stop the knife at the first point where there are at least $p$ whole MMS-rectangles to its left---the knife may need to move outside the cake in order for this to happen, as in Figure~\ref{fig:rectangle-2-3}(c) for any $p$. 
Consider two cases.

\emph{Case 1}: There are at least $q$ whole MMS-rectangles to the right of the knife.
Then, the knife indicates a vertical $p$:$q$-rectangle cut.
This is the case when $p=q=1$ in Figures \ref{fig:rectangle-2-3}(a), (b), and (d).

\emph{Case 2}: There are at most $q-1$ whole MMS-rectangles to the right of the knife. Then, by moving the knife slightly to the left, we obtain a cut for which there are at most $p-1$ MMS-rectangles entirely to its left, and at most $q-1$ MMS-rectangles entirely to its right. Therefore, at least $k-p-q+2$ MMS-rectangles must intersect the knife itself.
Since the knife width is $s$, these rectangles must lie in order vertically, with a vertical distance of at least~$s$ between consecutive rectangles.
Hence, they form a vertical $(k-p-q+2)$-rectangle stack.
This is the case when $p=q=1$ in Figure~\ref{fig:rectangle-2-3}(c).
\end{proof}

In the remainder of this section, 
given $y,y' \in [0,1]$ with $y\le y'$, 
we denote by $R(y,y')$ the rectangle $[0,1]\times [y,y']$.
We now prove a positive result for two agents matching the lower bound implied by \Cref{thm:impossibility}.

\begin{theorem}
\label{thm:rectangle-2}
For any land division instance with a rectangular land and $n=2$ agents, 
there exists an $s$-separated allocation 
in which each agent $i$ receives an axis-aligned rectangle of value at least $\emph{MMS}_i^3$.
\end{theorem}

\begin{proof} Call the agents Alice and Bob. Take a $1$-out-of-$3$ maximin partition of each agent, and consider two cases.

\emph{Case 1}: Both agents have a vertical $1$:$1$-rectangle cut. 
Assume without loss of generality that Alice's cut lies further to the left.
Give the rectangle to its left to Alice and the one to its right to Bob.
Then each agent receives an MMS-rectangle.

\emph{Case 2}: At least one agent, say Alice, has no vertical $1$:$1$-rectangle cut. By Lemma~\ref{lem:cuts}, she has a vertical $3$-rectangle stack, as in Figure~\ref{fig:rectangle-2-3}(c).
For the $i$-th rectangle in this stack (counting from the bottom), denote the $y$-coordinates of its top and bottom sides by $t_i$ and $b_i$, respectively.
Note that $t_1+s\leq b_2$ and $t_2+s\leq b_3$.

If Bob's value for $R(0,t_2)$ is at least $1$, then give $R(0,t_2)$ to Bob and $R(b_3,1)$ to Alice.
Otherwise, Bob values $R(0,t_2)$ less than $1$,
so his value for $R(b_2,1)$ is more than $2$; in this case, give $R(b_2,1)$ to Bob and 
$R(0,t_1)$ to Alice.
In both cases Alice's value is $1$ and the pieces are $s$-separated.
\end{proof}

For $n\geq 3$ agents, the analysis becomes more complicated. 
As in classic cake-cutting algorithms (e.g., \citep{DubinsSp61}), we would like to proceed recursively: give one agent a rectangle worth at least $1$, and divide the rest of the land among the remaining $n-1$ agents. In particular, for $n=3$, after allocating a piece to one agent, we would need to show that, for each of the remaining two agents, the rest of the land is worth at least $3$, so that we can apply Theorem~\ref{thm:rectangle-2}. In fact, to apply Theorem \ref{thm:rectangle-2}, we need an even stronger condition:
each agent should have three $s$-separated value-$1^+$ rectangles. 
However, the recursion step might yield a remainder land made of many chunks of such rectangles, 
where each chunk is worth less than $1$. 
We therefore need to adapt our definitions and lemmas accordingly.

\begin{definition}
For integers $p,q\geq 1$, a \emph{vertical $p$:$q$-value cut} of an agent is a rectangular strip of width $s$ such that the agent values the land on its left at least $p$ and the land on its right at least~$q$. 
\end{definition}

Note that for any $p$ and $q$, every vertical $p$:$q$-rectangle cut is also a vertical $p$:$q$-value cut, but the converse is not necessarily true.

For the following lemma, it is important that the agent's value function is normalized as explained earlier, i.e., the value of each MMS-rectangle is $1$ and the value outside the MMS-rectangles is $0$.
A \emph{land-subset} is a subset of the land after some pieces have possibly been allocated to other agents.

\begin{lemma}
\label{lem:cuts-fraction}
Consider an agent with a fixed $1$-out-of-$k$ maximin partition of the land, 
who takes part in a division of a rectangular land-subset.
Let $V\leq k$ be the agent's value for the land-subset.
For any integers $p,q \geq 1$ with $p+q \leq V$,
the agent has either a vertical $p$:$q$-value cut or a vertical 
$\lceil(\lfloor V\rfloor-p-q)/2\rceil$-rectangle stack.
\end{lemma}

\begin{proof}
Consider a vertical knife of width $s$ moving over the land-subset from left to right.
Denote by $v_L$ and $v_R$ the value of the land-subset to the left and to the right of the knife, respectively. 
As the knife moves, $v_L$ increases continuously from $0$ to $V$. 
Stop the knife at the first moment when $v_L = p$. 
There are two cases:

\emph{Case 1}:  $v_R \geq q$. Then, the knife indicates a vertical $p$:$q$-value cut.

\emph{Case 2}:  $v_R < q$. By moving the knife slightly to the left, we obtain a cut with $v_L<p$ and $v_R<q$. 
The value covered by the knife itself must then be more than $V-p-q$.
By the normalization assumption, all value comes from MMS-rectangles, and the value of each MMS-rectangle is $1$.
Hence, the value covered by the knife is made of `chunks' of value at most $1$ each, with a distance of at least $s$ between any two chunks.
Since the knife width is~$s$, these chunks must lie in order vertically, with a vertical distance of at least~$s$ between consecutive chunks.
We now show that these chunks can be grouped into $m$ $s$-separated rectangles of value at least $1$ each, where $m\geq \lceil(\lfloor V\rfloor-p-q)/2\rceil$.

Starting from the bottom, collect the chunks until the total value collected so far is at least $1$, and form a rectangle containing the collected chunks. Keep collecting chunks and forming rectangles in this way, until the total value of the remaining chunks is less than $1$. Let $m$ be the number of rectangles created in this manner. Since the value of each chunk is at most $1$, the value of each rectangle is less than $2$, while the remaining value is less than $1$. Therefore, the total value covered by the knife is less than $2 m + 1$. 
Hence, $2 m + 1 > V-p-q$. Since $m$ is an integer, this implies $2 m\geq \lfloor V\rfloor-p-q$ and therefore $m\geq \lceil(\lfloor V\rfloor-p-q)/2\rceil$.
\end{proof}

The following lemma establishes a weaker bound than \Cref{thm:rectangle-2} 
does; however, it applies to an arbitrary land-subset, and hence 
(unlike \Cref{thm:rectangle-2}) can
be used as part of a recursive argument. Its proof is essentially
identical to the proof of \Cref{thm:rectangle-2}: the only difference
is that we first look for a vertical $1$:$1$-value cut, 
and if we fail to find one, we invoke Lemma~\ref{lem:cuts-fraction}
to establish the existence of a vertical $3$-rectangle stack.

\begin{lemma}
\label{lem:rectangle-2-fraction}
Consider a rectangular land-subset and $n=2$ agents
who value it at least~$7$ each.
There is an $s$-separated allocation 
in which each agent receives an axis-aligned rectangle of value at least $1$.
\end{lemma}

\begin{proof}
We consider two cases.

\emph{Case 1}:  Both agents have a vertical $1$:$1$-value cut. 
Take the cut to the left, give the rectangle to its left to the cutter, 
and the rectangle to its right to the other agent. 

\emph{Case 2}:  At least one agent, say Alice, has no vertical $1$:$1$-value cut. 
By Lemma~\ref{lem:cuts-fraction}, she has a vertical $3$-rectangle stack,  
so we can proceed as in Case~2 of Theorem~\ref{thm:rectangle-2}.
\end{proof}

Let $\vreq(n)$ be the smallest value of $V$ such that if each of $n$ agents
values the land-subset at $V$ or higher, then there is always an $s$-separated allocation
of this land-subset in which each agent's value for her share is at least $1$.
Obviously $\vreq(1)=1$, and by Lemma~\ref{lem:rectangle-2-fraction} we know that
$\vreq(2)\leq 7$.
We can now provide a finite (exponential) MMS approximation for every~$n$. 

\begin{theorem}\label{thm:rectangle-n}
For any $n\geq 2$, given any land division instance with a rectangular land and $n$ agents,
there exists an $s$-separated allocation 
in which each agent $i$ receives an axis-aligned rectangle with value at least 
$\emph{MMS}_i^k$, where $k = \lceil 17\cdot 2^{n-3}\rceil$.
\end{theorem}

\begin{proof}
By our normalization, it suffices to prove that $\vreq(n)\le 17\cdot 2^{n-3}$ for all $n\geq 2$. 
We proceed by induction on~$n$. 
The base case $n\leq 2$ is handled by 
Lemma~\ref{lem:rectangle-2-fraction}.
We assume that the claim is true for $n-1$ and prove it for $n$.
Let $W := \vreq(n-1)$.
We will argue that $\vreq(n)\leq \max\{2 W, W+2n+4\}$.
From this, it follows that 
$\vreq(3)\leq 17$,
and
$\vreq(n)\leq 17\cdot 2^{n-3}$ for all $n\geq 3$, since $2n+4\le 17\cdot 2^{n-4}$ for all $n\ge 4$.

Given a rectangular land-subset and $n$ agents who value it at $V$ each, 
where $V \geq \max\{2 W, W+2n+4\}$, consider two cases.

\emph{Case 1}: All agents have a vertical $1$:$W$-value cut.
Implement the leftmost cut, give the rectangle to its left to the leftmost cutter (breaking ties arbitrarily), and divide the rectangle to its right among the remaining 
$n-1$ agents using the induction hypothesis.

\emph{Case 2}:  At least one agent, say Alice, has no vertical $1$:$W$-value cut. 
Then by Lemma~\ref{lem:cuts-fraction}, she has a vertical $d$-rectangle stack, where
\[
d
\geq \left\lceil \frac{\lfloor V\rfloor-W-1}{2}\right\rceil 
\geq 
\left\lceil \frac{\lfloor W+2n+4\rfloor -W-1}{2}\right\rceil = n+2.
\]

Denote the $y$-coordinates of the bottom sides (resp., top sides) 
of the value-$1^+$ rectangles in the stack, from bottom to top, by 
$(b_j)_{j\in [d]}$ (resp., $(t_j)_{j\in [d]}$).
Note that $t_j + s\leq b_{j+1}$ for all $j\in [d-1]$.
In the illustration below, $d=5$ and $s$ is at most the height of the white space between each pair of consecutive green rectangles.
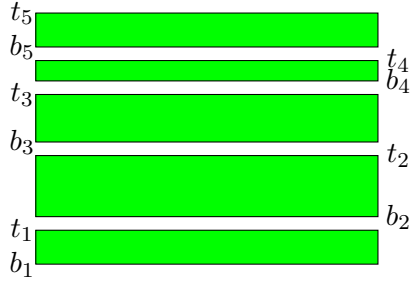
\begin{center}
\begin{tikzpicture}[scale=0.9]
\draw[draw=black,fill=green] (0,0) rectangle ++(5,0.5);
\draw[draw=black,fill=green] (0,0.7) rectangle ++(5,0.9);
\draw[draw=black,fill=green] (0,1.8) rectangle ++(5,0.7);
\draw[draw=black,fill=green] (0,2.7) rectangle ++(5,0.3);
\draw[draw=black,fill=green] (0,3.2) rectangle ++(5,0.5);

\node[draw=none] at (-0.2,0) {$b_1$};
\node[draw=none] at (5.3,0.7) {$b_2$};
\node[draw=none] at (-0.2,1.8) {$b_3$};
\node[draw=none] at (5.3,2.7) {$b_4$};
\node[draw=none] at (-0.2,3.2) {$b_5$};

\node[draw=none] at (-0.2,0.5) {$t_1$};
\node[draw=none] at (5.3,1.6) {$t_2$};
\node[draw=none] at (-0.2,2.5) {$t_3$};
\node[draw=none] at (5.3,3.0) {$t_4$};
\node[draw=none] at (-0.2,3.7) {$t_5$};
\end{tikzpicture}
\end{center}

For all $j\in[d]$,
let $H(j)$ be the set of agents (including Alice) who value $R(0,t_j)$ at least $W$.
We have $H(j)\subseteq H(j')$ whenever $j' > j$, and so
the sequence $(|H(j)|)_{j\in [d]}$ is non-decreasing.
Since $0\le |H(j)|\le n$ while $d \ge n+2$,
by the pigeonhole principle 
there exists a $z \in\{1,\dots,d-1\}$ such that $|H(z+1)|=|H(z)|$, and hence also $H(z+1)=H(z)$.

We consider three subcases based on the value of 
$|H(z)|=|H(z+1)|$. 
For brevity we set 
$R^-:=R(0, t_z)$ and $R^+:=R(b_{z+1}, 1)$.

\begin{itemize}
\item $|H(z)|=|H(z+1)|=0$. 
In this case,
give $R^-$ to Alice; she values it at least~$1$, since it contains her value-$1^+$ rectangle $R(b_1,t_1)$.
Divide $R^+$ among the other $n-1$ agents.
All agents value 
$R(0,t_{z+1})$ less than $W$,
so they value
$R^+$ at least $V-W\geq 2W-W = W$.
\item $|H(z)|=|H(z+1)| \in \{1,\ldots,n-1\}$. 
In this case, divide $R^-$ among the agents in $H(z)$ 
and divide $R^+$ among the remaining agents;
these latter agents value $R(0,t_{z+1})$ less than $W$, so they value 
$R^+$ more than $V-W \geq W$, since $V\geq 2W$.
\item $|H(z)|=|H(z+1)|=n$. 
Give $R^+$ to Alice; she values it at least $1$,
since it contains her value-$1^+$ rectangle $R(b_d,t_d)$.
Divide $R^-$ among the other $n-1$ agents, who value it at least $W$.
\end{itemize}
In all cases, each subdivision involves at most $n-1$ agents, and these agents value their 
land-subset at least $W$, so by the induction hypothesis each of them receives 
a value of at least~$1$.
\end{proof}

By adjusting the argument in the proof of Theorem~\ref{thm:rectangle-n}, 
we can obtain stronger bounds for $n=3$ and $n=4$; in particular,
we can guarantee each agent~$i$ 
a piece of value at least $\mms_i^{14}$ and $\mms_i^{24}$,
respectively.
The details can be found in Appendix~\ref{app:improved-3-4}.

The approximation factor in \Cref{thm:r-fat} for fat rectangles, which is linear in $n$, is much better than the factor in \Cref{thm:rectangle-n} for arbitrary rectangles, which is exponential in $n$. 
This raises the question of whether the MMS with respect to fat rectangles can itself be used as an approximation for the MMS with respect to general rectangles. 
In \Cref{app:square-mms-vs-rectangle-mms}
we show that the answer to this question is negative:
for any finite $r$,
the MMS with respect to $r$-fat rectangles does not provide any positive multiplicative approximation to the MMS with respect to general rectangles.

\section{Computing Maximin Allocations}
\label{sec:guillotine}

\subsection{Computational model}
\label{sec:computation-model}

The results in Section~\ref{sec:approx} are stated in terms of approximation guarantees.
To convert them into algorithms, we need to formally define our computational model. 
To do so, we propose a natural modification
of the classic (one-dimensional) cake-cutting model by \citet{RobertsonWe98} for the two-dimensional setting.

Consider an axis-aligned rectangle $L=[a_0, a_1]\times [b_0, b_1]$, which may be part of a larger land-subset.
We adapt the \cut{} and \eval{}
queries of the Robertson--Webb model to allow for horizontal and vertical cuts as follows.
The $\cut_i(|, L, \delta)$
query returns a value $a$ such that 
agent $i$ values the rectangle $[a_0, a]\times [b_0, b_1]$ at $\delta$, and 
the $\cut_i(-, L, \delta)$ 
query returns a value $b$ such that 
agent $i$ values the rectangle $[a_0, a_1]\times [b_0, b]$ at $\delta$;
we assume that this query returns `None'
if the agent values the entire rectangle less than $\delta$.
Similarly, the 
$\eval_i(|, L, a)$ query with 
$a_0\le a\le a_1$ returns the value that $i$ assigns to the rectangle 
$[a_0, a]\times [b_0, b_1]$,
whereas the
$\eval_i(-, L, b)$ query with 
$b_0\le b\le b_1$ returns the value that $i$ assigns to the rectangle 
$[a_0, a_1]\times [b_0, b]$.

We can now revisit the proofs of Theorems~\ref{thm:rectangle-2} and~\ref{thm:rectangle-n}
and check if they can be converted into algorithms that use $\cut$ and $\eval$ queries.
One can see that these proofs are constructive and their basic steps 
can be expressed in terms of these queries: 
a $p$:$q$-value cut can be implemented by two \cut{} queries,
and agents' values for rectangles of the form $R(x, y)$ can be determined
using \eval{} queries. 

However, these algorithms use the agents' 
$1$-out-of-$k$ maximin partitions as their starting points, and it is not clear if such partitions are efficiently computable.
Indeed, even in the one-dimensional case, there is no algorithm that always 
computes a maximin partition of an agent using finitely many queries, 
and the best known solution is an $\varepsilon$-additive approximation in time 
$O(n\log (1/\varepsilon))$ \citep[Thm.~3.2, Cor.~3.6]{ElkindSeSu22}.
For the two-dimensional case, even an $\varepsilon$-additive approximation
seems challenging.

\subsection{Guillotine partitions with separation}
\label{sec:guillotine-definition}

To circumvent this difficulty, we focus on maximin partitions with a special structure,
namely, {\em guillotine} partitions, which have been studied extensively in computational geometry
\citep{GonzalezRaSh94,ChristofidesHa95,AckermanBaPi06,Messaoud08,HorevKaKr09,AsinowskiBaMa14,RussoBoSf20}. 
 Formally, guillotine partitions are defined recursively, 
as follows.
\begin{definition}
Consider a land-subset $L=[a_0, a_1]\times [b_0, b_1]$,
a set of rectangles $\calP = \{P_1, \dots, P_t\}$, where $P_i\subseteq L$ 
for each $i\in [t]$, and a separation parameter $s$. 
We say that $\calP$ forms an {\em $s$-separated guillotine partition of $L$}
if one of the following three conditions holds:
\begin{itemize}
\item
$t=1$ and $P_1\subseteq L$;
\item
there exists an $a$ with $a_0<a< a_1-s$ and a partition of $\calP$
into two disjoint non-empty collections of rectangles $\calP_1$ and $\calP_2$
such that $\calP_1$ forms an $s$-separated guillotine partition of $[a_0, a]\times [b_0, b_1]$ 
and       $\calP_2$ forms an $s$-separated guillotine partition of $[a+s, a_1]\times [b_0, b_1]$;
\item
there exists a $b$ with $b_0<b< b_1-s$ and a partition of $\calP$
into two disjoint non-empty collections of rectangles $\calP_1$ and $\calP_2$
such that $\calP_1$ forms an $s$-separated guillotine partition of $[a_0, a_1]\times [b_0, b]$ 
and       $\calP_2$ forms an $s$-separated guillotine partition of $[a_0, a_1]\times [b+s, b_1]$.
\end{itemize}
\end{definition}

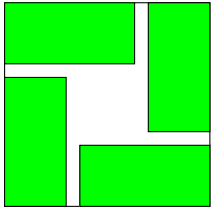
\begin{figure}
\centering
\begin{tikzpicture}[scale=0.9]
\draw (0,0) rectangle (3,3);
\draw [fill=green] (0,0) rectangle (0.9,1.9);
\draw [fill=green] (0,2.1) rectangle (1.9,3);
\draw [fill=green] (2.1,1.1) rectangle (3,3);
\draw [fill=green] (1.1,0) rectangle (3,0.9);
\end{tikzpicture}
\caption{Example of an $s$-separated partition
that is not a guillotine partition. The small space between each pair of `adjacent' rectangles has length $s$.}
\label{fig:non-guillotine}
\end{figure}

Intuitively, an $s$-separated guillotine partition is obtained by a sequence of cuts, 
where each cut splits a rectangle into two $s$-separated rectangles. 
All partitions in Figure~\ref{fig:rectangle-2-3} are guillotine partitions, 
while Figure~\ref{fig:non-guillotine} provides an example of an $s$-separated partition
that is not a guillotine partition.

\subsection{Maximin allocations with guillotine partitions}
\label{sec:MMS-guillotine}

Let $\gumms_i^{k,s}$ denote the maximin share of agent~$i$ with respect to $s$-separated guillotine partitions into $k$ parts (it is defined similarly to the $\mms$ in \Cref{def:MMS}, except that the supremum is over all $s$-separated guillotine partitions $\mathbf{P}$).
We will show that, for every constant $k$, it is possible to approximate $\gumms_i^{k,s}$ up to an additive factor of $\varepsilon$, in time $O((\log(1/\varepsilon))^k)$.

We define a \emph{guillotine tree} to be a rooted binary tree in which each non-leaf node is labelled with either `H' (for `Horizontal') or `V' (for `Vertical'). 
Each guillotine tree with $k$ leaf nodes represents a way in which a rectangular land-subset can be partitioned into $k$ parts using guillotine cuts.\footnote{The number of such trees has been studied by \citet{YaoChCh03} and \citet{AckermanBaPi06}.} 
For example, the tree in \Cref{fig:guillotine-tree} represents partitions into $k = 4$ parts in which the rectangle is first cut horizontally, then the top part is cut horizontally and the bottom part is cut vertically.

\begin{figure}
\begin{center}
\begin{tikzpicture}
\draw (3.67,3) -- (4.33,4.5) -- (6,6) -- (7.67,4.5) -- (8.33,3);
\draw (4.33,4.5) -- (5,3);
\draw (7.67,4.5) -- (7,3);
\draw[fill=white] (6,6) circle [radius=0.3];
\draw[fill=white] (4.33,4.5) circle [radius=0.3];
\draw[fill=white] (7.67,4.5) circle [radius=0.3];
\draw[fill=white] (3.67,3) circle [radius=0.3];
\draw[fill=white] (5,3) circle [radius=0.3];
\draw[fill=white] (7,3) circle [radius=0.3];
\draw[fill=white] (8.33,3) circle [radius=0.3];
\node at (6,6) {\small $H$};
\node at (4.33,4.5) {\small $H$};
\node at (7.67,4.5) {\small $V$};
\end{tikzpicture}
\end{center}
\caption{
Example of a guillotine tree for a partition into $k = 4$ parts.
\label{fig:guillotine-tree}
}
\end{figure}
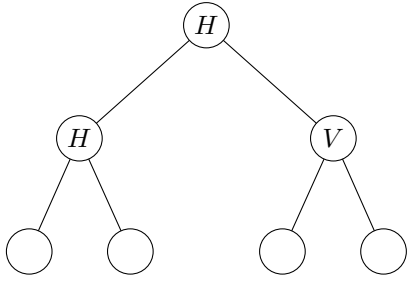

Our algorithm works by checking all guillotine trees with $k$~leaves and finding one that yields an approximately optimal partition.
We will need the following lemma as a subroutine.
Given a value function and a tree~$T$, denote by $T$-$\text{MMS}^{k,s}$ the maximin share with respect to $s$-separated partitions into $k$ parts according to the tree $T$.

\begin{lemma}
\label{lem:guillotine-recurse}
Let $k$ be a constant.
Given a rectangular land-subset~$L$, a value function~$v$ such that $v(L)=1$, a guillotine tree $T$ with $k$ leaves, and real numbers $r\ge 0$ and $\varepsilon > 0$, it is possible to decide in time $O((\log(1/\varepsilon))^{k-1})$ which of the following two options holds:
\begin{itemize}
\item $T$-$\emph{MMS}^{k,s} \ge r$;
\item $T$-$\emph{MMS}^{k,s} < r + k\varepsilon$.
\end{itemize}
(At least one option must hold; if both options hold, it suffices to return one of them.)
If we return the first option, we can output an $s$-separated partition such that each part has value at least $r$ as well.
\end{lemma}

\begin{proof}
We proceed by induction on $k$.
The base case $k = 1$ holds trivially, as no cut needs to be made.
Let $k\ge 2$. Assume without loss of generality that the first cut according to $T$ is vertical, and the left and right subtree has $k_0$ and $k-k_0$ leaves, respectively, for some $k_0\in \{1,\ldots,k-1\}$.
For each vertical cut $C$, denote by $\text{Left}(C)$ the rectangle to the left of~$C$ and by $\text{Right}(C)$ the rectangle to its right.

We place a cut $C$ of width~$s$ such that 
$v(\text{Left}(C)) = v(L)/2$, and recurse on both sides using the inductive hypothesis.
Denote by $T_0$ and $T_1$ the left and right subtree of $T$, respectively, and let $\beta_0 := T_0\text{-MMS}^{k_0,s}(\text{Left}(C))$ and $\beta_1 := T_1\text{-MMS}^{k-k_0,s}(\text{Right}(C))$.

\emph{Case 1}: 
$\beta_0 \ge r$
and $\beta_1 \geq r$. Then we return that $T$-$\text{MMS}^{k,s} \ge r$.

\emph{Case 2}:  
$\beta_0 < r + k_0\varepsilon$
and $\beta_1 < r + (k-k_0)\varepsilon$.
In this case, both values are less than $r + k\varepsilon$.
Moving the cut to the left will only make the MMS on the left smaller, and moving the cut to the right will only make the MMS on the right smaller. 
Therefore, there is no guillotine partition according to~$T$ with MMS at least $r + k\varepsilon$, so we return that $T$-$\text{MMS}^{k,s} < r + k\varepsilon$.

\emph{Case 3}: 
Either $\beta_0 < r + k_0\varepsilon$ while $\beta_1 \geq r$, or $\beta_0 \ge r$ while $\beta_1 < r + (k-k_0)\varepsilon$.
In the former case, we place a cut $C'$ of width $s$ such that $v(\text{Left}(C')) = 3\cdot v(L)/4$ (that is, the average between $v(L)/2$ and $v(L)$), and repeat the same procedure.
Analogously, in the latter case,
we place a cut $C'$ of width $s$ such that $v(\text{Left}(C')) =  v(L)/4$ (that is, the average between $v(L)/2$ and $0$), and repeat the same procedure.

We continue until either we return that $T$-$\text{MMS}^{k,s} \ge r$ or that $T$-$\text{MMS}^{k,s} < r + k\varepsilon$, or we have placed two cuts, $C_1$ and $C_2$, such that
\begin{enumerate}[label=(\roman*)]
\item $C_1$ is to the left of $C_2$, and 
$v(\text{Left}(C_2))  - v(\text{Left}(C_1)) < \varepsilon$;
\item 
$T_0\text{-MMS}^{k_0,s}(\text{Left}(C_1)) < r + k_0\varepsilon$
(while $T_1\text{-MMS}^{k-k_0,s}(\text{Right}(C_1)) \geq r$);
\item 
$T_1\text{-MMS}^{k-k_0,s}(\text{Right}(C_2)) < r + (k-k_0)\varepsilon$
(while $T_0\text{-MMS}^{k_0,s}(\text{Left}(C_2)) \geq  r$).
\end{enumerate}
Note that if we have not terminated, we will reach this situation after placing $O(\log(1/\varepsilon))$ cuts.
By conditions~(i) and (ii), 
$T_0\text{-MMS}^{k_0,s}(\text{Left}(C_2)) < (r + k_0\varepsilon)+\varepsilon \le r + k\varepsilon$.
Moreover, by condition~(iii), 
$T_1\text{-MMS}^{k-k_0,s}(\text{Right}(C_2)) < r + k \varepsilon$, too.
Hence, we may return that $T$-$\text{MMS}^{k,s} < r + k\varepsilon$ in this situation.

Each recursive call takes time $O((\log(1/\varepsilon))^{k-2})$ by the inductive hypothesis, and the number of recursive calls is $O(\log(1/\varepsilon))$, so the total running time is $O((\log(1/\varepsilon))^{k-1})$.
In the case where we return $T$-$\text{MMS}^{k,s} \ge r$, our procedure allows us to construct the associated partition as well.
\end{proof}

We are now ready to prove our main theorem.

\begin{theorem}
\label{thm:guillotine-tree}
Let $k$ be a constant.
Given a rectangular land-subset $L$, a value function~$v$ normalized so that $v(L)=1$, and a real number $\varepsilon > 0$, we can compute an $s$-separated guillotine partition of $L$ into $k$ rectangles such that the value of each part is at least $\gumms^{k,s} - \varepsilon$ in time $O((\log(1/\varepsilon))^{k})$.
\end{theorem}

\begin{proof}
Let $B$ be the interval consisting of possible values of $\gumms^{k,s}$; initially, $B = [0,1]$.
While the length of $B = [x,y]$ is greater than $\varepsilon$, we proceed as follows.
Let $z$ be the midpoint of $B$.
Applying \Cref{lem:guillotine-recurse} over all (constant number of) guillotine trees~$T$ with $k$ leaves and using $\varepsilon/(4k)$ instead of $\varepsilon$ in the lemma statement, we can determine in time $O((\log(4k/\varepsilon))^{k-1}) = O((\log(1/\varepsilon))^{k-1})$ which of the following two options holds:
\begin{enumerate}[label=(\roman*)]
\item $T$-$\text{MMS}^{k,s} \ge z$ for some $T$;
\item $T$-$\text{MMS}^{k,s} < z + \varepsilon/4$ for all $T$.
\end{enumerate}

If (i) holds, we can shrink $B$ to $[z,y]$, and by \Cref{lem:guillotine-recurse}, we can also find an $s$-separated partition such that each part has value at least $z$.
On the other hand, if (ii) holds, we can shrink $B$ to $[x,z+\varepsilon/4]$.
We repeat this procedure as long as the length of~$B$ is greater than $\varepsilon$.
In each iteration, the length of $B$ is reduced to at most $3/4$ times the previous length, and we have an $s$-separated partition such that each part has value at least the smallest value in $B$.\footnote{This condition holds trivially at the beginning because the smallest value in $B$ is $0$.}
After $O(\log(1/\varepsilon))$ iterations, the length of $B$ becomes at most $\varepsilon$, and we output the smallest value in $B$ along with the associated partition.
The guarantee follows from the definition of $B$, and the running time is $O((\log(1/\varepsilon))^{k-1}\cdot \log(1/\varepsilon)) = O((\log(1/\varepsilon))^{k})$.
\end{proof}

\begin{remark}
\Cref{lem:guillotine-recurse} assumes that the usable shapes are arbitrary rectangles. 
However, it can also be adapted to the setting in which the usable shapes are squares (or, more generally, fat rectangles).
The only place in the proof that needs to be modified is the induction base. 
For $k=1$, given a rectangular land-subset $L$, instead of evaluating $L$ and comparing its value to $r$, we compute (approximately) the highest value of a square contained in $L$, and compare it to $r$.

To compute the highest value of a square contained in $L$, partition $L$ along its longer side into $d:=\ceil{2/\varepsilon}$ strips of value at most $\varepsilon/2$. 
Assume without loss of generality that the longer side is parallel to the $y$ axis.
Let $y_1,\ldots,y_{d-1}$ be the $y$-coordinates of the cuts, and let $x_0$ be the leftmost coordinate of the rectangle. For each $j\in[d-1]$, evaluate the square with lower-left corner at $(x_0,y_j)$ and side length equal to the shorter side of $L$, provided that the square is contained in $L$.
Moreover, evaluate the topmost and bottommost rectangles contained in $L$ with side length equal to the shorter side of $L$.
Of the (at most $d+1$) resulting squares, let $R^*$ be a square with the highest value.

If $v(R^*)\geq r$, then we can return that $T$-$\text{MMS}^{1,s} \ge r$.
Suppose that $v(R^*)< r$.
By construction, for every square $R\subseteq L$, there is an evaluated square $R'\subseteq L$ such that $v(R\setminus R')\leq \varepsilon/2$.
Therefore, $v(R) - v(R^*) \leq \varepsilon/2$, and so $v(R) \le r+\varepsilon/2 < r + \varepsilon$.
Hence, we can return that $T$-$\text{MMS}^{k,s} < r + \varepsilon = r+k\varepsilon$. 

The induction step of  \Cref{lem:guillotine-recurse} remains unchanged.
Moreover, once we plug the adapted \Cref{lem:guillotine-recurse} into \Cref{thm:guillotine-tree}, its proof works as is. 
Therefore, 
both \Cref{lem:guillotine-recurse} and \Cref{thm:guillotine-tree} hold whenever $U$ is the family of squares, fat rectangles, or arbitrary rectangles, with an extra factor of $O(1/\varepsilon)$ in the running time due to the induction base.
\end{remark}

\subsection{General partitions}
\label{sec:MMS-general}

How much value do we lose by considering guillotine partitions instead of general ones?
Figure~\ref{fig:non-guillotine} illustrates that this loss is non-trivial.  
To obtain a bound on the approximation ratio, we relate it to the \emph{guillotine separation} problem studied in computational geometry.

\begin{definition}
Let $S$ be a set of pairwise-disjoint geometric objects,
and $P$ be a partition of the plane.
The \emph{$S$-value} of $P$ is the number of parts of $P$ that contain at least one whole object from $S$.

The \emph{guillotine-separation number} of $S$ is the largest $S$-value of a ($0$-separated) guillotine partition.
\end{definition}
In general, the guillotine-separation number of $S$ might be smaller than $|S|$. 
For example, if $S$ is the set of rectangles in Figure~\ref{fig:non-guillotine}, then since every guillotine partition must cut through at least one rectangle, the guillotine-separation number is at most $3$ (it is easy to see that it is exactly $3$).

\begin{definition}
Let $U$ be a family of geometric objects and $m \geq 2$ an integer. The \emph{$m$-th guillotine-separation number} of $U$, denoted $\gsn(U,m)$, is the smallest guillotine-separation number of a subset of $m$ objects from $U$.
\end{definition}
The first in-depth study of guillotine-separation numbers was done by \citet{PachTa00}, who proved upper and lower bounds for $\gsn(U,m)$ for various families $U$. \citet{AbedChCo15} proved that
\[
\frac{m}{2} + o(m)
\geq
\gsn(\text{axis-parallel unit squares},m) \geq  \frac{m}{2}
\]
and
\[
\gsn(\text{axis-parallel squares},m)\geq \frac{m}{81}.
\]
\citet{KhanPi20} improved the lower bound to 
\begin{align}
\gsn(\text{axis-parallel squares},m)\geq \frac{m}{40},
\label{eq:KhanPi-squares}
\end{align}
and proved that 
\begin{align}
\gsn(\text{axis-parallel rectangles},m)\geq \frac{m}{1+\log_2 m} = \frac{m}{\log_2(2m)}
\label{eq:KhanPi-rectangles}
\end{align}
and
\[
\gsn(\text{axis-parallel rectangles},m)\geq \frac{m}{\log_3 (m+2)}.
\]
The exact guillotine-separation numbers for squares and rectangles are still open.

We now relate the guillotine-separation number to the maximin share.

\begin{theorem}
\label{thm:guillotine-vs-general}
Fix the family $U$ of usable pieces (either axis-parallel squares or axis-parallel rectangles).

Let $\gumms_i^{k,s}$ denote the maximin share of agent~$i$ with respect to $s$-separated guillotine partitions into $k$ parts. 
Let $m$ be an integer for which $\gsn(U,m)\geq k$. 
It holds that
$\gumms_i^{k,s}\geq \emph{MMS}_i^{m,s}$.
\end{theorem}
\begin{proof}
Let $S$ be the set of rectangles in a 1-out-of-$m$ $s$-separated maximin partition of agent~$i$.
Replace each rectangle $Q\in S$
with $\wrap{Q}{s}$ (see \Cref{thm:r-fat}).
Since the original rectangles are $s$-separated, the wrapped rectangles are pairwise-disjoint 
(see Figure~\ref{fig:square-wrap}).

Since $\gsn(U,m)\geq k$, there exists a guillotine partition in which at least $k$ parts contain a whole rectangle from $S$; denote these whole rectangles by $\wrap{Q_1}{s},\allowbreak\ldots,\allowbreak\wrap{Q_k}{s}$.
The rectangles $Q_1,\ldots,Q_k$ are $s$-separated, and the value of each of them is at least $\text{MMS}_i^{m,s}$. 
Therefore, $\gumms_i^{k,s}\geq \text{MMS}_i^{m,s}$, as claimed.
\end{proof}

Using the aforementioned lower bounds \eqref{eq:KhanPi-squares} and \eqref{eq:KhanPi-rectangles} of \citet{KhanPi20}, we get:
\begin{corollary}
\label{cor:guillotine-vs-general}
(a)
When $U$ is the family of axis-parallel squares, for every integer $k\geq 2$,
we have
$\gumms_i^{k,s}\geq \mms_i^{40k,s}$.

(b)
Let $h(m) := m/\log_2(2m)$.
When $U$ is the family of axis-parallel rectangles, for every integer $k\geq 2$,
$\gumms_i^{k,s}\geq \mms_i^{\lceil h^{-1}(k)\rceil,s}$.
\end{corollary}
Note that\footnote{See https://math.stackexchange.com/a/4546052.} the inverse of the function $h$ satisfies $h^{-1}(k) \in O(k\log k)$.

Combining \Cref{cor:guillotine-vs-general}
with \Cref{thm:guillotine-tree}, we obtain that, for any constant $k$, we can compute $\mms_i^{40k,s}$ (for squares) or $\mms_i^{\lceil h^{-1}(k)\rceil,s}$ (for rectangles) up to an additive error of $\varepsilon$ in time $O((\log(1/\varepsilon))^{k+1})$ (for squares) or $O((\log(1/\varepsilon))^{k})$ (for rectangles).
We can then apply \Cref{thm:rectangle-n} to obtain the following corollary.
\begin{corollary}
Let $n\geq 2$ be a constant, and
let $m := \lceil h^{-1}(17\cdot 2^{n-3})\rceil$, where $h$ is defined as in \Cref{cor:guillotine-vs-general}.
For any $\varepsilon>0$,
we can compute 
in time polynomial in $\log(1/\varepsilon)$
an $s$-separated allocation among $n$ agents, in which each agent~$i$ receives an axis-parallel rectangle with value at least $\emph{MMS}_i^{m, s}-\varepsilon$.
\end{corollary}
\begin{proof}
We first use Theorem~\ref{thm:guillotine-tree} to compute for each agent~$i$ an $s$-separated guillotine partition into $k$ parts such that $i$'s value for each part is at least $\gumms_i^{k,s} - \varepsilon$ in time 
$O((\log(1/\varepsilon))^k)$, where $k = 17\cdot 2^{n-3}$.
We normalize each agent's value for each of her MMS-rectangles to $1$ and the value outside these rectangles to $0$, and apply Theorem~\ref{thm:rectangle-n} on the resulting instance.\footnote{We assume here that we can make queries on the normalized instance. (Alternatively, this can be simulated by normalizing the values of the MMS-rectangles so that the least valuable rectangle has value~$1$, shrinking the remaining MMS-rectangles until they also have value~$1$, and extending the \cut{} query in our modified Robertson--Webb model to finding a value $a$ such that the \emph{union} of rectangles $[a_0,a]\times([b_0,b_1]\cup [b_2,b_3]\cup\dots\cup[b_{2t},b_{2t+1}])$ has value $\delta$, for any given $a_0,b_0,b_1,\dots,b_{2t+1},\delta$.) Also, we use the value $17\cdot 2^{n-3}$ instead of $\vreq(n)$ in the implementation of the algorithm from Theorem~\ref{thm:rectangle-n}.}
Theorem~\ref{thm:rectangle-n} guarantees that agent~$i$ receives value at least $\gumms_i^{k,s} - \varepsilon$.
By \Cref{cor:guillotine-vs-general}(b), this value is at least $\mms_i^{m,s} - \varepsilon$.
\end{proof}

Similarly, using \Cref{thm:r-fat} and \Cref{cor:guillotine-vs-general}(a), we get:
\begin{corollary}
Let $n\geq 2$ be a constant, and
let $m := 40\cdot(4n-5)$.
For any $\varepsilon>0$,
we can compute 
in time 
polynomial in $\log(1/\varepsilon)$
an $s$-separated allocation among $n$ agents, in which each agent~$i$ receives an axis-parallel square with value at least $\emph{MMS}_i^{m, s}-\varepsilon$.
\end{corollary}

In \Cref{app:alternative-MMS}, we present two alternative algorithms for approximating the maximin share. These algorithms are based on discretization: we partition the land into horizontal and vertical strips of value at most $\varepsilon$, and find a maximin partition with respect to the rectangles with corners at the grid points. This technique works only when the usable shapes are the rectangles. For any integer $k\geq 2$ and real number $\varepsilon>0$, we show how to compute the following in time polynomial in $k$ and $1/\varepsilon$:
\begin{itemize}
\item A guillotine partition into $k$ rectangles, each of which has a value of at least $\gumms_i^{k,s}-\varepsilon$ (\Cref{app:alternative-guillotine});
\item A general partition into $k$ rectangles, each of which has a value of at least $\mms_i^{3k,s}-\varepsilon$ 
(\Cref{app:alternative-general}). This algorithm uses a recent algorithm for the problem of \emph{maximum independent set of rectangles} \citep{galvez20223}.
\end{itemize}

\section{Conclusion and Future Work}

In this paper, we continue the quest of bringing the theory of fair division closer to practice by investigating fair land allocation under separation constraints.
We establish bounds on achievable maximin share guarantees for a variety of shapes and develop a number of new techniques in the process; some of our techniques may be independently interesting from a computational geometry perspective.
While our maximin share bound for fat rectangles is polynomial in the number of agents~$n$ (\Cref{thm:r-fat}), the one for arbitrary rectangles is exponential (\Cref{thm:rectangle-n}).
Improving the latter bound, perhaps to a polynomial bound, is a challenging question which will likely require novel geometric insights.
For future work, it would also be interesting to test our algorithms on real land division data (see, e.g., the work of \citet{ShtechmanGoSe22}) as well as to explore the possibilities of efficient computation with non-guillotine or other types of cuts.
Finally, one could try to extend our results to the more general setting where different agents may have unequal entitlements to the land.\footnote{The case of unequal entitlements has been previously studied for both divisible \citep{Segalhalevi19,CsehFl20} and indivisible resources \citep{ChakrabortyIgSu21,ChakrabortyScSu21,ChakrabortySeSu22,SuksompongTe22}.}

\section*{Acknowledgments} 

This work was partially supported by the European Research Council (ERC) under grant number 639945 (ACCORD), by the Israel Science Foundation under grant number 712/20, by the Singapore Ministry of Education under grant number MOE-T2EP20221-0001, and by an NUS Start-up Grant.
We would like to thank Kshitij Gajjar for his insights and references regarding guillotine partitions, 
Alex Ravsky for his insights regarding rainbow independent sets, 
Qiaochu Yuan for mathematical help, 
and the anonymous reviewers of the 30th International Joint Conference on Artificial Intelligence (IJCAI 2021) and Computational Geometry:~Theory and Applications for their valuable comments and suggestions.

\bibliographystyle{plainnat}
\bibliography{main,geometry}

\appendix

\section{Improved Bounds for Three and Four Agents}
\label{app:improved-3-4}

In the case of three agents, Theorem~\ref{thm:rectangle-n} implies a guarantee of $\mms_i^{17}$.
We improve this guarantee in the following theorem.

\begin{theorem}
\label{thm:rectangle-3}
For every land division instance with a rectangular land and $n=3$ agents,
there exists an $s$-separated allocation 
in which each agent $i$ receives an axis-aligned rectangle with value at least 
$\emph{MMS}_i^{14}$.
\end{theorem}

\begin{proof}
Fix a $1$-out-of-$14$ maximin partition for each agent. We consider two cases.

\emph{Case 1}:  All agents have a vertical $1$:$3$-rectangle cut. 
Implement the leftmost such cut and give the rectangle to its left to the leftmost cutter, 
breaking ties arbitrarily.
Each of the other two agents has at least three MMS-rectangles in the remaining land, so we can divide the land between them using Theorem~\ref{thm:rectangle-2}.

\emph{Case 2}:  At least one agent, say Alice, has no vertical $1$:$3$-rectangle cut. 
Then by Lemma~\ref{lem:cuts}, she has a vertical $d$-rectangle stack, where $d\ge 12$.

Denote the $y$-coordinates of the top sides (resp., bottom sides) of the rectangles in the stack, from bottom to top, by $(t_j)_{j\in [d]}$ (resp., $(b_j)_{j\in [d]}$).
Note that $t_j + s\leq b_{j+1}$ for all $j\in [d-1]$.

For each $j\in[d]$, 
the rectangle $R(0,t_j)$ contains the first $j$ value-$1^+$ rectangles in the stack counting from the bottom. 
Denote by $H(j)$ the set of agents whose value for $R(0,t_j)$ is at least $\vreq(2)$. 

We have $H(j)\subseteq H(j')$ whenever $j' > j$, and so
the sequence $(|H(j)|)_{j\in [d]}$ is non-decreasing.
Since $0\le |H(j)|\le 3$ while $d \ge 12 > 5$,
by the pigeonhole principle 
there exists a $z \in[d-1]$ such that $|H(z+1)|=|H(z)|$, and hence also $H(z+1)=H(z)$.
Crucially, the agents not in $H(z+1)=H(z)$ value $R(0,t_{z+1})$ less than $\vreq(2)$, so they value $R(b_{z+1},1)$ more than $14-\vreq(2)\ge 7$. 

We consider four subcases based on the value of $|H(z)|=|H(z+1)|$. For readability we set 
$R^-:=R(0, t_z)$ and $R^+:=R(b_{z+1}, 1)$.

\begin{itemize}
\item $|H(z)|=|H(z+1)|=0$. 
Then all agents value $R^+$ more than $7$.
Give $R^-$ to Alice (who values it at least $1$), and divide $R^+$ between the other two agents using 
Lemma~\ref{lem:rectangle-2-fraction}.
\item $|H(z)|=|H(z+1)|=1$. 
Then give $R^-$ to the agent in $H(z)$ (who values it at least~$7$), and divide $R^+$ between the other two agents,
who value it more than~$7$ each, using Lemma~\ref{lem:rectangle-2-fraction}. 
(Note that this is where we need $14$ in the theorem statement.)
\item $|H(z)|=|H(z+1)|=2$.
Then divide $R^-$ between the two agents in $H(z)$ using 
Lemma~\ref{lem:rectangle-2-fraction}, and give $R^+$ to the agent not in $H(z)$,
who values it more than $14-7 > 1$.
\item $|H(z)|=|H(z+1)|=3$. 
Give $R^+$ to Alice (who values it at least $1$), and divide $R^-$ between the other two agents (who value it at least $\vreq(2)$ each) using Lemma~\ref{lem:rectangle-2-fraction}.\qedhere
\end{itemize}
\end{proof}

We now proceed to four agents.
Recall that, by the proof of \Cref{thm:rectangle-n}, $\vreq(3)\leq 17$.

\begin{theorem}
\label{thm:rectangle-4}
For any land division instance with a rectangular land and $n=4$ agents,
there exists an $s$-separated allocation 
in which each agent $i$ receives an axis-aligned rectangle with value at least 
$\emph{MMS}_i^{24}$.
\end{theorem}

\begin{proof}
We consider two cases.

\emph{Case 1}:  All agents have a vertical $3$:$3$-rectangle cut. 
Cut the cake at a median of the four cuts. Divide the rectangles to the left (resp., right) between the two leftmost (resp., rightmost) cutters using Theorem \ref{thm:rectangle-2}.

\emph{Case 2}:  At least one agent, say Alice, has no vertical $3$:$3$-rectangle cut. 
Then by Lemma \ref{lem:cuts}, she has a vertical $d$-rectangle stack, where $d\geq 20 > 16$.

For $1\leq j\leq d$, define $t_j$ and $b_j$ in the same way as in Theorem~\ref{thm:rectangle-3}, and for $z\in\{2,3\}$, denote by $H(j,z)$ the set of agents who value $R(0,t_j)$ at least $\vreq(z)$.
Note that $|H(j,z)|$ is non-decreasing in $j$ and non-increasing in $z$.

Similarly to Theorem \ref{thm:rectangle-3}, we look for equal consecutive values of $H$, i.e., values $j$ such that $H(j,z)=H(j+1,z)$. 
In the following claims, our goal is to partition the land into two $s$-separated rectangles such that some $\ell\in[n-1]$ agents value one rectangle at least $\vreq(\ell)$ and the other $n-\ell$ agents value the other rectangle at least $\vreq(n-\ell)$, so that the instance reduces into two smaller instances.

\begin{claim}
\label{claim:rectangle-1}
If $|H(j,z)| = |H(j+1,z)| = z$ for some $1\leq j\leq d-1$ and $z\in\{2,3\}$, then we are done.
\end{claim}
\begin{proof}
The condition implies that $H(j,z) = H(j+1,z)$.
The $z$ agents in $H(j,z)$ value $R(0,t_j)$ at least $\vreq(z)$. 
The other $n-z$ agents are not in $H(j+1,z)$, so they value $R(0,t_{j+1})$ and therefore $R(0,b_{j+1})$ less than $\vreq(z)$; this implies that they value $R(b_{j+1}, 1)$ more than $24-\vreq(z)$. 
One can check that $24-\vreq(z) \geq \vreq(n-z)$ for all $1\leq z\leq n-1$.
We can therefore allocate $R(0,t_j)$ to the first $z$ agents and 
$R(b_{j+1}, 1)$ to the remaining $n-z$ agents.
\end{proof}

\begin{claim}
\label{claim:rectangle-2}
If $|H(j,z)| = |H(j+1,z)| = z-1$ for some $1\leq j\leq d-1$ and $z\in\{2,3\}$, then we are done.
\end{claim}
\begin{proof}
Similarly to Claim~\ref{claim:rectangle-1}, 
the $z-1$ agents in $H(j,z)$ value $R(0,t_j)$ at least $\vreq(z) \ge \vreq(z-1)$, and the other $n-z+1$ agents value $R(b_{j+1}, 1)$ more than $24-\vreq(z)$. 
We have
$24 - \vreq(2)\geq 24-7 = 17 \geq \vreq(3)$, so $24-\vreq(z) \geq \vreq(n-z+1)$ for all $z\in\{2,3\}$.
We can allocate $R(0,t_j)$ to the $z-1$ agents in $H(j,z)$ and 
$R(b_{j+1}, 1)$ to the remaining $n-z+1$ agents.
\end{proof}

\begin{claim}
\label{claim:rectangle-3}
If $|H(j,2)| = |H(j+1,2)| = 0$ for some $1\leq j\leq d-1$, then we are done.
\end{claim}
\begin{proof}
All $n$ agents value $R(0,t_{j+1})$ and hence $R(0,b_{j+1})$ less than $\vreq(2)$,
so they value $R(b_{j+1}, 1)$ more than $24-\vreq(2)\geq \vreq(3)$.
We allocate $R(0,t_j)$ to Alice, who values it at least $1$, and 
$R(b_{j+1}, 1)$ to the other three agents.
\end{proof}

\begin{claim}
\label{claim:rectangle-4}
If $|H(j,3)| = |H(j+1,3)| = 4$ for some $1\leq j\leq d-1$, then we are done.
\end{claim}
\begin{proof}
All $n$ agents value $R(0,t_{j})$ at least $\vreq(3)$.
We can allocate $R(b_{j+1},1)$ to Alice, who values it at least $1$, and 
$R(0,t_j)$ to the remaining three agents.
\end{proof}

We are now ready to complete the proof of \Cref{thm:rectangle-4}.
Since $|H(j,z)|$ has only $n+1=5$ possible values and $d\geq 16$, there must be a set $H(j,2)$ repeated at least four times in succession, i.e., there exists a $j_2$ such that $H(j_2,2)=H(j_2+1,2)=H(j_2+2,2)=H(j_2+3,2)$.
We denote by $m$ the size of this repeating set (if there is more than one such set, we pick one arbitrarily). We proceed in cases.
\begin{itemize}
\item If $m=0$, then by Claim~\ref{claim:rectangle-3} we are done.
\item If $m=1$, then by Claim~\ref{claim:rectangle-2} we are done.
\item If $m=2$, then by Claim~\ref{claim:rectangle-1} we are done.
\item If $m =  3$ or $4$, then we proceed by 
the value of 
$|H(j_2+1,3)|$:
\begin{itemize}
\item If $|H(j_2+1,3)| \leq 2$, then 
$H(j_2+1,3)\subsetneq H(j_2+1,2) = H(j_2,2)$.
Give $R(0,t_{j_2})$ to a pair of agents in 
$H(j_2,2)$, ensuring that this pair includes all of the agents in $H(j_2+1,3)$.
Give $R(b_{j_2+1},1)$ to the other two agents, who are not in 
$H(j_2+1,3)$ and therefore value 
$R(b_{j_2+1},1)$ at least $24-\vreq(3)\geq 24-17 = 7 \geq \vreq(2)$.
\item If $|H(j_2+1,3)| \geq 3$, then
the sequence $|H(j_2+1,3)|,~|H(j_2+2,3)|,~|H(j_2+3,3)|$, which is non-decreasing, 
contains only the numbers $3$ and $4$.
Hence, one of these numbers must appear at least twice in succession.
If this repeating number is $3$, then by Claim~\ref{claim:rectangle-1} we are done.
Else, this number is $4$, and by Claim~\ref{claim:rectangle-4} we are again done.
\qedhere
\end{itemize}
\end{itemize}
\end{proof}

\section{Rectangle MMS vs. Fat-Rectangle MMS}
\label{app:square-mms-vs-rectangle-mms}

For the following discussion, we focus on the value function of a single agent and assume that the land is the unit square $[0,1]\times[0,1]$.
Denote by $\mms^k[r]$ the agent's $1$-out-of-$k$ MMS with respect to $r$-fat rectangles, where $k\ge 1$ is an integer and $r\ge 1$ a real number; in particular, $\mms^k[1]$ is the MMS with respect to squares.
Denote by $\mms^k[\infty]$ the agent's MMS with respect to arbitrary rectangles. 
Obviously, $\mms^k[1]\leq \mms^k[r]\leq \mms^k[\infty]$ for any $r\ge 1$.

When $s=0$, \citet{SegalhaleviNiHa17} proved that there always exists a collection of $k$ non-overlapping squares such that the value of each square is at least $1/(2k)$ times the value of the land. 
Since $\mms^k[\infty]$ is at most $1/k$ times the value of the land, we have that
$\mms^k[r]\geq \frac{1}{2}\cdot\mms^k[\infty]$
for any $r\ge 1$.
However, the following proposition shows that a similar result does not hold when $s>0$, even for the case $k = 2$.

\begin{proposition}
\label{prop:square-mms-vs-rectangle-mms}
Assume that the land is the unit square $[0,1]\times[0,1]$.
For any $r\geq 1$, $\varepsilon>0$, and $s\in (0,1)$, there exists a value function such that 
$\emph{MMS}^2[r]\leq 2\varepsilon\cdot  \emph{MMS}^{2}[\infty].$
\end{proposition}

\begin{proof}
The claim holds trivially for $\varepsilon \geq 1/2$, so we assume 
$\varepsilon < 1/2$, which implies $\varepsilon < r$.

Recall that we measure distance in $\ell_\infty$ norm.
If the statement is true for $\varepsilon = \varepsilon'$ for some value $\varepsilon'$, it is also true for all $\varepsilon > \varepsilon'$.
Hence, it suffices to prove the statement when $\varepsilon$ is sufficiently small so that
 $s + 2s\varepsilon/r < 1$.
 
The value is uniformly distributed in the two strips shown in \Cref{fig:square-mms-vs-rectangle-mms}, whose value is $1$ each.
Specifically, the bottom strip is the rectangle $[0,s]\times [0,s\varepsilon/r]$ and the top strip is the rectangle $[0,s]\times [s+s\varepsilon/r,s+2s\varepsilon/r]$.

Since the distance between the two strips is exactly $s$, we have $\mms^2[\infty]=1$.
On the other hand, 
$\mms^2[r]\leq 2\varepsilon$.
To see this, consider two $s$-separated $r$-fat rectangles, and assume that each of them has a positive value. 
Since each strip has width $s$ and height $s\varepsilon/r < s$, the two rectangles cannot overlap the same strip.
Hence, one of the rectangles---call it $R$---must overlap only the bottom strip.
If $R$ has height greater than $2s\varepsilon/r$, then its distance to any point in the top strip is less than $s$, which is impossible as the other $r$-fat rectangle must overlap the top strip.
So the height of $R$ is at most $2s\varepsilon/r$.
Since $R$ is $r$-fat, its width is at most $2s\varepsilon$, which is a $2\varepsilon$ fraction of the width of the bottom strip.
This means that the value of $R$ is at most $2\varepsilon$, and therefore $\mms^2[r]\leq 2\varepsilon$.
\end{proof}

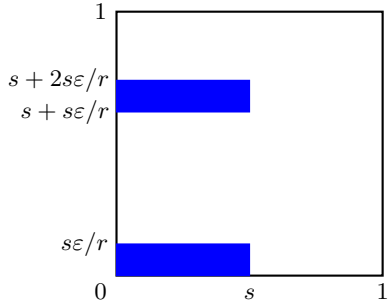
\begin{figure}
\begin{center}
\begin{tikzpicture}[scale=0.7]
\draw[black, thick] (0,0) rectangle (5,5);
\draw[blue, fill=blue] (0,0) rectangle (2.5,0.6);
\draw[blue, fill=blue] (0,3.1) rectangle (2.5,3.7);
\node at (-0.3,-0.3) {\footnotesize $0$};
\node at (-0.3,5) {\footnotesize $1$};
\node at (5,-0.3) {\footnotesize $1$};
\node at (2.5,-0.35) {\footnotesize $s$};
\node at (-0.6,0.6) {\footnotesize $s\varepsilon/r$};
\node at (-0.97,3.1) {\footnotesize $s+s\varepsilon/r$};
\node at (-1.09,3.7) {\footnotesize $s+2s\varepsilon/r$};
\end{tikzpicture}
\end{center}
\caption{
Construction in the proof of \Cref{prop:square-mms-vs-rectangle-mms}.
}
\label{fig:square-mms-vs-rectangle-mms}
\end{figure}

\section{Alternative Algorithms for MMS Approximation}
\label{app:alternative-MMS}

\subsection{Guillotine partitions}
\label{app:alternative-guillotine}

The following theorem shows that we can compute a nearly optimal $s$-separated
guillotine maximin partition efficiently. 
Our algorithm proceeds by discretizing the land and finding an optimal 
$s$-separated guillotine partition that is consistent with this
discretization---such a partition can be computed by dynamic programming.

\begin{theorem}
\label{thm:guillotine-dp}
Consider a rectangular land-subset $L$, 
an agent $i$ with $v_i(L)=1$, a separation parameter $s>0$, and a positive integer~$k$.
Let $V := \gumms_i^{k,s}$.
Then, given $\varepsilon>0$, $s$, and $k$,
we can compute an $s$-separated guillotine partition of $L$
into $k$ parts such that $i$'s value for each part is at least $V-\varepsilon$,
in time polynomial in $k$ and $1/\varepsilon$.
\end{theorem}

\begin{proof} The algorithm proceeds in two steps.
~
\paragraph{Step 1: Discretization.}
Assume that $L=[a_0, a_1]\times [b_0, b_1]$.
Let $\delta = \varepsilon/4$, and let $d= \left\lceil 1/\delta\right\rceil$.
Using $\cut$ queries, we  cut $L$ into $d$ vertical strips
such that $i$'s value for each strip is at most $\delta$; let $x_0=a_0, x_1, \dots, x_{d-1}, x_d=a_1$
be the respective cut points.
Similarly, using $\cut$ queries, we can cut $L$ into $d$ horizontal strips
such that $i$'s value for each strip is at most $\delta$; let $y_0=b_0, y_1, \dots, y_{d-1}, y_d=b_1$
be the respective cut points.

Let $G=\{(x_i, y_j): 0\le i, j\le d\}$ be the set of \emph{grid-points}, i.e., intersection points of the cuts.
Define a \emph{grid-rectangle} as a rectangle whose corners are in $G$, that is, a rectangle $[x_i, x_{i'}]\times [y_j, y_{j'}]$ for some $0\le i< i'\le d$ and $0\le j< j'\le d$. We now show that, by limiting our attention to grid-rectangles, we lose a value of at most $\varepsilon$.

\begin{claim}
\label{claim:guillotine-to-general}
For each $s$-separated guillotine partition with minimum value $V$, there exists an $s$-separated guillotine partition such that all parts are grid-rectangles, and the value of every part is at least $V-\varepsilon$.
\end{claim}

\begin{proof}[Proof of Claim~\ref{claim:guillotine-to-general}]
We first prove that, for every axis-parallel rectangle $R\subseteq L$, there exists a grid-rectangle $I(R)\subseteq R$ such that $v_i(I(R))\geq v_i(R)-\varepsilon$.
Consider a rectangle $R = [x, x']\times [y, y']$. Find the smallest value of $i$
such that $x_i\ge x$ and the largest value of $i'$ such that $x_{i'}\le x'$.
Similarly, find the smallest value of $j$ such that 
$y_j\ge y$ and the largest value of $j'$ such that $y_{j'}\le y'$.
Denote the grid-rectangle $[x_i, x_{i'}]\times [y_j, y_{j'}]$ by $I(R)$.
By construction, $v_i(I(R))\geq v_i(R)-4\delta=v_i(R)-\varepsilon$, since we removed four `borders' of value at most $\delta$ each.

Now, construct a new partition in which each rectangle $R$ in the original partition is replaced by its subset grid-rectangle $I(R)$.
Clearly, the new partition is still $s$-separated, still a guillotine partition, contains only grid-rectangles, and the value of every part in this partition is at least $V-\varepsilon$.
\end{proof}

\paragraph{Step 2: Dynamic program.}
For each 4-tuple of indices $i, i', j, j'$ with $0\le i< i'\le d$, 
$0\le j<j'\le d$, and each $t\in [k]$, let $B[i, i', j, j', t]$
be the maximum value $z$ such that there exists an $s$-separated guillotine partition of 
$[x_i, x_{i'}]\times [y_j, y_{j'}]$ into $t$ grid-rectangles in which
the value of each rectangle is at least $z$.
We can compute the $O(k/\varepsilon^4$) entries of $B$ by dynamic programming as follows.

For each $0\le i< i'\le d$ and $0\le j< j'\le d$, we set 
$B[i, i', j, j', 1]$ to the value of the rectangle $[x_i, x_{i'}]\times [y_j, y_{j'}]$.
To compute $B[i, i', j, j', t]$ with $t>1$, we consider the following options:
\begin{itemize}
\item[(i)] 
splitting $[x_i, x_{i'}]\times [y_j, y_{j'}]$ horizontally, i.e., picking 
some  integer $t'$ with $0<t'<t$, as well as some integer $\ell$ with $j<\ell<j'$, finding the smallest value 
$\ell'$ with $\ell<\ell'<j'$ such that $y_{\ell'}-y_\ell\ge s$,
and computing $\min\{B[i, i', j, \ell, t'], B[i, i', \ell', j', t-t']\}$; and 
\item[(ii)] 
splitting $[x_i, x_{i'}]\times [y_j, y_{j'}]$ vertically, i.e., picking 
some  integer $t'$ with $0<t'<t$, as well as some integer $r$ with $i<r<i'$, finding the smallest value 
$r'$ with $r<r'<i'$ such that $x_{r'}-x_r\ge s$,
and computing $\min\{B[i, r, j, j', t'], B[r', i', j, j', t-t']\}$.
\end{itemize}
The value $B[i, i', j, j', t]$ is obtained by considering
all $O(k/\varepsilon$) solutions of this form (iterating over all  choices of $t'$ and all possible values of 
$\ell$ and $r$ in (i) and (ii)) and finding the best one;
we set $B[i, i', j, j', t]=0$ if no such split is feasible.

Now, suppose that $B[0, d, 0, d, k] = z$, and let $\calP$
be the associated partition, which can be computed by standard
dynamic programming techniques \citep[Chapter~15]{CormenLeRi09}. By construction, $\calP$
is an $s$-separated guillotine partition, 
$i$'s value for each rectangle in $\calP$
is at least $z$,
and this is the highest minimum-value attainable by an $s$-separated guillotine partition into grid-rectangles.
By \Cref{claim:guillotine-to-general}, we have $z\ge V-\varepsilon$.

It follows that,
in the partition output by the dynamic program, agent $i$ 
values each rectangle at least $V-\varepsilon$.
\end{proof}

\subsection{General partitions}
\label{app:alternative-general}

Next, we present an alternative algorithm for finding an approximate maximin partition, that does not involve guillotine partitions. 
Instead, we use an approximation algorithm for the problem of \emph{maximum independent set of rectangles (MISR)}. An instance of MISR consists of a set $A$ of axis-parallel rectangles. 
The goal is to find a largest subset of $A$ in which all rectangles are pairwise-disjoint. This is an NP-hard problem, but it has various polynomial-time approximation algorithms. 
The most recent one is a 3-factor approximation algorithm by \citet{galvez20223}. 

\begin{lemma}
\label{lem:3-factor-decision}
Given a rectangular land-subset, a value function $v$, an integer $k\geq 2$, and real numbers $r\ge 0$ and $\varepsilon>0$, it is possible to decide in time polynomial in $k$ and $1/\varepsilon$ which of the following two options holds:
\begin{itemize}
\item $\mms^{k,s} \ge r$;
\item $\mms^{3 k,s} < r + \varepsilon$.
\end{itemize}
(At least one option must hold; if both options hold, it suffices to return one of them.)
If we return the first option, we can output an $s$-separated partition into $k$ rectangles such that each part has value at least $r$ as well.
\end{lemma}
\begin{proof}
We perform the discretization step as in the proof of \Cref{thm:guillotine-dp}---this yields $O(1/\varepsilon^2)$ grid-points and therefore $O(1/\varepsilon^4)$ grid-rectangles. 
Let $Z$ be the set of all grid-rectangles with value at least $r$. 
Replace each rectangle $R\in Z$ with $\wrap{R}{s}$ defined in the proof of \Cref{thm:r-fat}, and apply the algorithm of \citet{galvez20223} to $Z$.

\emph{Case 1}:  The algorithm finds a subset $Z'\subseteq Z$ with at least $k$ rectangles.
By removing the wrapping from each rectangle in $Z'$, we get an $s$-separated partition with $k$ rectangles of value at least $r$ each. Hence, we can return that $\mms^{k,s} \ge r$ along with the partition.

\emph{Case 2}:  The algorithm finds a subset with fewer than $k$ rectangles.
Due to the approximation guarantee of the algorithm, any subset of pairwise-disjoint rectangles in~$Z$ contains fewer than $3k$ rectangles.
This implies that there is no subset of $3k$ $s$-separated grid-rectangles with value at least~$r$ each.
By \Cref{claim:guillotine-to-general}, 
there is no subset of $3k$ $s$-separated rectangles with value at least $r+\varepsilon$ each.
Hence, we can return that $\mms^{3 k,s} < r + \varepsilon$.
\end{proof}

\begin{theorem}
\label{thm:3-factor}
Given a rectangular land-subset $L$, a value function $v$ normalized so that $v(L)=1$, an integer $k\geq 2$, and a real number $\varepsilon>0$, it is possible to compute in time polynomial in $k$ and $1/\varepsilon$ an $s$-separated partition into $k$ rectangles with value at least  $\mms^{3 k,s} - \varepsilon$ each.
\end{theorem}
\begin{proof}
We proceed by binary search.
Specifically, we maintain an interval $I =[r_1,r_2]$ such that $\mms^{k,s} \ge r_1$ and $\mms^{3k,s} < r_2+\varepsilon/2$.
Initially, we set $I := [0,1]$; since $v(L)=1$, this choice indeed satisfies the two inequalities.
Then, we iterate the following procedure:
\begin{itemize}
\item Let $r$ be the midpoint of $I$;
\item Apply the algorithm of \Cref{lem:3-factor-decision} (using $\varepsilon/2$ instead of $\varepsilon$);
\item If the algorithm returns that $\mms^{k,s} \ge r$, then set the left endpoint of $I$ to $r$, and keep the resulting partition into $k$ rectangles;
\item Otherwise (i.e., the algorithm returns that $\mms^{3k,s} < r+\varepsilon/2$), set the right endpoint of $I$ to $r$.
\end{itemize}

After $O(\log(1/\varepsilon))$ iterations, we have an interval of length at most $\varepsilon/2$; denote it by $[r, r+d]$, where $d\le \varepsilon/2$.
By construction, $\mms^{3k,s} < (r+d)+\varepsilon/2 \le r+\varepsilon$,
and we have an $s$-separated partition into $k$ rectangles of value at least $r > \mms^{3k,s}-\varepsilon$.
\end{proof}
\end{document}